\documentclass[format=acmsmall,screen,natbib=false,nonacm]{acmart}

\usepackage{mathtools}
\usepackage[T1]{fontenc}
\usepackage{bm}
\usepackage{bbm}

\usepackage[capitalise,nameinlink,noabbrev]{cleveref}
\usepackage[backend=biber, style=acmnumeric, maxbibnames=4, maxalphanames=4, backref=true]{biblatex}
\usepackage{xargs}
 
\usepackage{tikz}
\usetikzlibrary{positioning, quotes}
\usepackage{mdframed}

\def \R{\mathbb R}
\def \N{\mathbb N}
\newcommand{\sset}[1]{\left\{ #1\right\}}
\newcommand{\ssets}[1]{\{ #1\}}
\newcommand{\fwh}[1]{\; \left| \; #1 \right.}

\newcommand{\union}{\cup}
\newcommand{\map}{\longrightarrow}
\newcommand{\inters}{\cap}   
\newcommand{\then}{\Longrightarrow}
\newcommand{\vecc}{\bm}
\DeclareMathOperator*{\expectation}{\mathbb E}
\newcommand{\expect}[2][]{\expectation_{#1}\nolimits\left[#2\right]}

\renewcommand\vec{\bm}

\newcommand{\noparticipation}{\ensuremath{\varnothing}}
\newcommand{\revenue}{\ensuremath{\mathsf{Rev}}}
\newcommand{\welfare}{\ensuremath{\mathsf{SW}}}

\newcommand{\unitmatrix}{\ensuremath{\bm{1}}}
\newcommand{\zeromatrix}{\ensuremath{\bm{0}}}

\newcommand\downward[1]{{}^{\downharpoonleft}\negthinspace #1}
\newcommand\upward[1]{{}^{\upharpoonleft}\negthinspace #1}

\AtEndPreamble{%
\theoremstyle{acmdefinition}
\newtheorem{remark}[theorem]{Remark}}

\begin{document}

\title{Discrete Single-Parameter Optimal Auction Design}
\titlenote{A preliminary version of this paper, not including any proofs, the application~\cref{sec:application}, and the generalization of~\cref{lemma:integral-optimality} to TDI systems, appeared in the proceedings of SAGT 2024~\parencite{gh2024_sagt}.}

\author{Yiannis Giannakopoulos}
\affiliation{
  \institution{University of Glasgow}
  \country{United Kingdom}}
\email{yiannis.giannakopoulos@glasgow.ac.uk}
\orcid{0000-0003-2382-1779}

\author{Johannes Hahn}
\affiliation{
  \institution{University of Technology Nuremberg}
  \country{Germany}}
\email{johannes.hahn@utn.de}
\orcid{0009-0005-3016-2996}

\begin{abstract}
We study the classic single-item auction setting of Myerson, but under the
assumption that the buyers' values for the item are distributed over
\emph{finite} supports. Using strong LP duality and polyhedral theory, we
rederive various key results regarding the revenue-maximizing auction, including
the characterization through virtual welfare maximization and the optimality of
deterministic mechanisms, as well as a novel, generic equivalence between
dominant-strategy and Bayesian incentive compatibility. 

Inspired by this, we abstract our approach to handle more general auction
settings, where the feasibility space can be given by arbitrary convex
constraints, and the objective is a linear combination of revenue and social
welfare. We characterize the optimal auctions of such systems as generalized
virtual welfare maximizers, by making use of their KKT conditions, and we
present an analogue of Myerson's payment formula for general discrete
single-parameter auction settings. Additionally, we prove that total
unimodularity of the feasibility space is a sufficient condition to guarantee
the optimality of auctions with integral allocation rules.

Finally, we demonstrate this KKT approach by applying it to a setting where
bidders are interested in buying feasible flows on trees with capacity
constraints, and provide a combinatorial description of the (randomized, in
general) optimal auction.
\end{abstract}

\maketitle

\newpage
\section{Introduction}
The design of optimal auctions~\parencite{Krishna2009a,Milgrom:2004aa} that
maximize the seller's revenue is a cornerstone of the field of mechanism design
(see, e.g., \parencite[Ch.~9]{Jehle2001a} and \parencite{Hartline2007a}),
established into prominence by the highly-influential work
of~\textcite{Myerson1981a}, and traced back to the seminal work
of~\textcite{Vickrey1961a}.

In its most classical form~\parencite{Myerson1981a}, which is the basis for the
setting we are studying in our paper as well, there is a single item to sell and
the problem is modelled as a Bayesian game. The seller has only incomplete
information about the bidders' true valuations of the item, in the form of
independent (but not necessarily identical) probability distributions; these
distributions are assumed to be public knowledge across all participants in the
auction. The players/bidders submit bids to the auctioneer/seller and the seller
decides (a) who gets the item, and with what probability (since lotteries are
allowed), and (b) how much the winning bidders are charged for this transaction.

In this game formulation, the strategies of the players are the different bids
they can submit, and it could well be the case that bidders misreport their true
valuations, if this can result in maximizing their own personal utility.
Therefore, a desirable feature of mechanism design in such settings is the
implementation of auctions which provably guarantee that truth-telling is an
\emph{equilibrium} of the game; such auctions are called \emph{truthful} (or
\emph{incentive compatible (IC)}). Perhaps surprisingly, the celebrated
Revelation Principle of \textcite{Myerson1981a} ensures that restricting our
attention within the class of such well-behaved selling mechanisms is without
loss for our purposes.

The seminal work of~\textcite{Myerson1981a} provides a complete and
mathematically satisfying characterization of revenue-maximizing truthful
auctions in the aforementioned single-item setting, under the assumption that
valuation/bidding spaces are \emph{continuous}. It explicitly constructs an
optimal auction that (a) is deterministic, i.e.\ the item is allocated to a
single bidder (with full probability), or not sold at all, (b) satisfies
truthfulness in a very strong sense, namely under dominant-strategy equilibrium,
and not just in-expectation (see~\cref{sec:truthfulness} for more details), and
(c) has a very elegant description, enabled via the well-known \emph{virtual
valuation} ``trick'' (see~\eqref{definition:virtual_values}); this casts the
problem into the domain of welfare-maximization, simplifying it significantly by
stripping away the game-theoretic incentives components, and transforming it to
a ``purely algorithmic'' optimization problem --- resembling the familiar, to
any computer scientist, notion of a reduction (a formalization of this
connection, even for more general environments, can be found in the work
of~\textcite{Cai2012b,Cai2013a}).

Still, the assumption of continuity may be considered as too strong for many
practical, and theoretical, purposes. Any conceivable instantiation of an
auction on a computing system will require some kind of discretization; not only
as a trivial, unavoidable consequence of the fundamentally discrete nature of
computation (i.e., ``bits''), but also for practical reasons: bids are usually
expected to be submitted as increments of some common denomination (e.g.,
``cents''). And any implementation of optimal auction design as an optimization
problem, would need to be determined by finitely many parameters and variables,
to be passed, e.g., to some solver. Furthermore, although many of the key
properties and results for the continuous setting can be derived as a limiting
case of a sequence of discrete ones, in general the opposite is not true: most
of the techniques used in traditional auction theory rely on real analysis and
continuous probability, thus breaking down when called to be applied to discrete
spaces.

The above reasons highlight the importance of deriving a clear and robust theory
of optimal auction design, under the assumption of \emph{finite} value spaces.
In other words, a discrete analogue of Myerson's~\parencite{Myerson1981a}
theory. During the last couple of decades, various papers within the field of
algorithmic game theory have dealt with this task; see \cref{sec:related-work}
for a more detailed overview. 
Our goal in this paper is to first rederive existing key results, in a unified
way, with an emphasis on clarity, simplicity, and rigorousness; and, do this via
purely discrete optimization tools (namely, LP duality and polyhedral
combinatorics), ``agnostically'', rather than trying to mimic and discretize
Myerson's~\parencite{Myerson1981a} approach for the continuous setting.
Secondly, this comprehensiveness and transparency allows us to lift our approach
up to handle quite general single-parameter mechanism design environments, by
concisely formulating our problem as an elegant KKT system.

\subsection{Related Work}
\label{sec:related-work}

To the best of our knowledge, the first to explicitly study optimal auction
design in a discrete setting were~\textcite{Bergemann2007} and
\textcite{Elkind2007}; the latter offers a more complete treatment, providing a
natural discretization of Myerson's~\parencite{Myerson1981a} techniques,
including ``ironing'' of non-regular distributions (see~\cref{sec:ironing}). A
limitation of~\parencite{Elkind2007} is that it establishes that the discrete
analogue of Myerson's auction is optimal within the more restrictive class of
dominant-strategy incentive compatible (DSIC) mechanisms, instead of using the
standard, weaker notion of Bayesian incentive compatibility (BIC).

In a discussion paper, \textcite{Malakhov2004} study discrete auction
environments with identical bidders under BIC, providing a simpler, equivalent
characterization of truthfulness, through a set of local constraints. We will
make critical use of this characterization, appropriately adapted to our
general, non-symmetric setting of our paper (see~\cref{sec:IC-locality}). The
treatment of \parencite{Malakhov2004}  puts emphasis on linear programming (LP)
formulations, and derive an interesting, flow-based description of optimality
for general, multi-dimensional mechanism design settings; the monograph
of~\textcite{Vohra2011a} provides a comprehensive treatment of this approach.

All aforementioned approaches work, essentially, by adapting the key steps of
Myerson's derivations, from the continuous to the discrete setting.
\Textcite{Cai2019} provide a totally different, and very powerful, approach
based on Lagrangian duality. Conceptually, their paper is clearly the closest to
ours. \Parencite{Cai2019} followed a line of work, where duality proved very
useful in designing optimal multiple-item auctions in the continuous case (see,
e.g., \parencite{Daskalakis:2017aa,gk2014}). Although the duality framework of
\parencite{Cai2019} is fundamentally discrete, it was also designed for
multi-dimensional revenue-maximization, a notoriously difficult and complex
problem. Therefore, its instantiation for a single-parameter Myersonian setting
(see~\parencite[Sec.~4]{Cai2019}) results, arguably, in a rather involved
presentation. One of the goals of our paper is exactly to demystify duality for
single-item domains, by making use of classical LP duality,
particularly tailored for our problem, instead of the more obscure Lagrangian
flows interpretation in~\parencite{Cai2019}, resulting in greater transparency
and a wider spectrum of questions that we can attack (see~\cref{sec:Myerson}).

\subsection{Our Results}
\label{sec:results}

We begin our presentation by introducing our \emph{single-parameter} auction
design setting, and fixing some overarching notation, in~\cref{sec:prelims}. Our
model formulation is deliberately general, allowing for arbitrary feasibility
domains $\mathcal{A}$ for the auction's allocation; we will specialize this to
the standard distributional simplex when studying the classical Myersonian
single-item setting in~\cref{sec:Myerson}, however we want to be able to capture
the abstract convex environments we study later in~\cref{sec:KKT}.
Importantly, in~\cref{sec:truthfulness} we discuss in detail the two different
notions of truthfulness used for our problem, and in~\cref{sec:IC-locality} we
provide a local characterization of truthfulness, essentially proved
in~\textcite{Malakhov2004}, which we will extensively use in our optimization
formulation throughout our paper.

\Cref{sec:Myerson} includes our rederivation of the key components of
Myerson's~\parencite{Myerson1981a} theory for single-item revenue-maximization,
but for finite-support distributions, as well as some novel results. They all
arise, in a unified way, through a chain of traditional LP duality, presented
in~\cref{sec:duality-chain-myerson} (see~\cref{fig:chain_tikz} for a concise
pictorial view). The resulting revenue-maximizing auction, together with some
key results characterizing optimality, are given in the ``master''
~\cref{th:discrete-optimal-myerson}: in a nutshell, the optimal auction first
transforms the submitted bids to \emph{virtual} bids and then \emph{irons} them,
finally allocating the item to the highest non-negative (virtual, ironed) bidder.
Similar to the classical results of~\parencite{Myerson1981a} for continuous
domains, this auction turns out to be deterministic and truthful in the
strongest DSIC sense, ``for free'', although we are optimizing within the much
wider space of lotteries under BIC. To the best of our knowledge,
Point~\ref{item:optimal-single_item-BICvsDIC} where we formalize the equivalence
of DSIC and BIC, under revenue-maximization, as a more fundamental and general
consequence of the polyhedral structure of our feasibility space, rather than
just a feature of the particular optimal auction solution format, is novel. The
remaining subsections~\ref{sec:determinism-vs-randomization-Myerson},
\ref{sec:DSIC-vs-BIC-Myerson}, and \ref{sec:ironing}, are dedicated to
elaborating and formally proving the various components
of~\cref{th:discrete-optimal-myerson}. A point worth noting is that our virtual
value~\eqref{eq:virtual_values-def} and ironing~\eqref{eq:ironed_virtual_values}
transformations are not ``guessed'' and then proven to impose optimality, as is
the case with prior work in the area, but rather arise organically as a
necessity of our strong LP duality technique.

Inspired by the transparency of our duality framework in~\cref{sec:Myerson}, we
try to generalize our approach to a more general single-parameter mechanism
design setting, where the feasibility space $\mathcal A$ is given by arbitrary
convex constraints, and the optimization objective is a linear combination of
revenue and social welfare; see~\cref{sec:notation-KKT}. Our results are
summarized in master~\cref{th:discrete-optimal-general}, which is essentially
the analogue of~\cref{th:discrete-optimal-myerson}. Given the generality of our
model in this section, we have to depart from our basic LP duality tools
of~\cref{sec:Myerson}, and make use of the more general KKT conditions
framework, including duality and complementary slackness; our KKT formulation is
discussed in~\cref{sec:KKT-model}. 
The abstraction of our model allows for a
very concise description of the optimal auction's allocation and
payment rules (see~\cref{sec:general-virtual-welfare-maximization}). Similarly
to the single-item setting of~\cref{sec:Myerson}, we can again show that
optimizing under the more restrictive notion of \ref{eq:DSIC-def} truthfulness
is without loss for our optimization objective. Furthermore, we investigate
under what structural conditions of our underlying feasibility space we can
``generically'' guarantee that there exists an optimal auction that does not
need to allocate fractionally/randomly, i.e.\ it is integral; it turns out, that
\emph{total unimodularity} is such a sufficient condition
(see~\cref{sec:integral-auctions-KKT} for more details and definitions).

It is important to point out here that, in principle, one could derive the main
results of~\cref{sec:Myerson} for the single-item case by making use of the more
general KKT setting of~\cref{sec:KKT}. In other words,
\cref{th:discrete-optimal-myerson} can be viewed as a special case
of~\cref{th:discrete-optimal-general}. Nevertheless, we deliberately choose in
our paper to first, independently develop the special, single-item theory
of~\cref{sec:Myerson}, not as much as a warm-up for the conceptually more
demanding and abstract presentation of~\cref{sec:KKT}, but also for essential
technical reasons: many components of our proofs from~\cref{sec:Myerson} are
needed in order to keep the technical difficulty of~\cref{sec:KKT} manageable,
i.e., our paper is built in a modular way so that we do not unnecessarily repeat
technical parts from the single-item to the general case, but at the same time
those parts are key components for the way our proofs are presented for the
general case in~\cref{sec:KKT}. Additionally, if one was to actually rederive
our results and presentation for the single-item case as a special-case
instantiation of the more general framework in~\cref{sec:KKT}, this would result
in a very hard to penetrate presentation, obscuring from the key insights and
clarity provided by the traditional LP tools used
in~\cref{sec:Myerson}.  

Finally, in~\cref{sec:application} we demonstrate the transparency and strength
of our framework, by applying it to a capacitated tree setting, inspired by
real-life gas network structures~\parencite{GasLib}, where each bidder wants to
send flow between a fixed origin-destination pair. In addition to the KKT
formulation of the problem, we ``unravel'' its optimal solution (as dictated by
Point~\ref{item:optimal-single_parameter-general-explicit} of
\cref{th:discrete-optimal-general}) to derive a purely combinatorial,
algorithmic description of the allocation and payment rules
(see~\cref{sec:combinatorial-algo-tree-flows}), that reveals an interesting
economics interpretation of \emph{edge pricing}.

\section{Preliminaries}
\label{sec:prelims}

\subsection{Model and Notation}
\label{sec:model-notation}
We use $\R$, $\R_+$, and $\N$, for the set of reals, non-negative reals, and
non-negative integers, respectively. For any positive integer $k$ we denote
$[k]\coloneqq\ssets{1,2,\dots,k}$. 

\paragraph{Single-parameter settings} In a (Bayesian) single-parameter auction
design setting there are $n\geq 1$ bidders, and each bidder $i\in[n]$ has a
value $v_i\in\R_+$ for being allocated a single ``unit'' of some ``service''.
Each value $v_i$ is drawn independently from a distribution (with cdf) $F_i$
with support $V_i\subseteq\R_+$, called the \emph{prior} of bidder $i$. We will
use $f_i$ to denote the probability mass function (pmf) of $F_i$. These
distributions are public knowledge, however the realization $v_i$ is private
information of bidder $i$ only. In this paper we only study \emph{discrete}
auction settings, where the prior supports $V_i$ are \emph{finite}. For
notational convenience, we denote the corresponding product distribution of the
\emph{value profiles} $\vec{v}=(v_1,v_2,\dots,v_n)\in \vec{V}\coloneqq
\times_{i=1}^n V_i$ by $\vec{F}\coloneqq\times_{i=1}^n F_i$, and we also use
$\vec{V}_{-i}\coloneqq\times_{j\in[n]\setminus{i}} V_j$ and
$\vec{F}_{-i}\coloneqq\times_{j\in[n]\setminus{i}} F_j$.

There is also a set of feasible outcomes $\mathcal A\subseteq\R_+^n$, each
outcome $\vec{a}=(a_1,a_2,\dots,a_n)\in\mathcal A$ corresponding to bidder $i$
being allocated a ``quantity'' $a_i$. Throughout this paper we assume that
$\mathcal A$ is \emph{convex}. A canonical example is the classical single-item
auction setting (which we study in \cref{sec:Myerson}), where $a_i$ can be
interpreted as the probability of a lottery assigning the item to bidder $i$, in
which case the feasibility set $\mathcal{A}$ is the $n$-dimensional simplex
$\mathcal S_n\coloneqq\sset{\vec{a}\in\R^n_+\fwh{\sum_{i=1}^n a_i \leq 1}}$.

\paragraph{Auctions} An auction $M=(\vec{a},\vec{p})$ consists of an allocation
rule $\vec{a}:\vec{V}\map\mathcal A$ and a payment rule
$\vec{p}:\vec{V}\map\R^n$ that, given as input a vector of bids $\vec{b}\in
\vec{V}$, dictates that each bidder $i$ should get allocated quantity $a_i(\vec
b)$ and submit a payment of $p_i(\vec{b})$ to the auctioneer. 

Given such an auction $M$, the (ex-post) utility of a bidder $i$, when their
true value is $v_i\in V_i$ and bidders submit bids $\vec{b}\in\vec{V}$, is
\begin{equation}
\label{eq:utility_ex-post_definition}
   u_i^M(\vec b; v_i) = u_i(\vec b; v_i) \coloneqq a_i(\vec b)\cdot v_i - p_i(\vec b). 
\end{equation}
Using the distributional priors $F_i$ to capture the uncertainty about other
bidders' behaviour, we can also define the \emph{interim} utility of a bidder,
when having true value $v_i\in V_i$ and bidding $b_i\in V_i$ as
\begin{equation*}
\label{eq:utility_interim_definition}
    U_i(b_i; v_i) 
    \coloneqq \expect[\vec{b}_{-i}\sim \vec{F}_{-i}]{u_i(b_i,\vec{b}_{-i};v_i)}
    = A_i(b_i)\cdot v_i - P_i(b_i), 
\end{equation*}
where
\[
A_i(b_i) \coloneqq \expect[\vec{b}_{-i}\sim \vec{F}_{-i}]{a_i(b_i,\vec{b}_{-i})}
\qquad \text{and} \qquad
P_i(b_i) \coloneqq \expect[\vec{b}_{-i}\sim \vec{F}_{-i}]{p_i(b_i,\vec{b}_{-i})}
\]
are the interim versions of the allocation and payment rules of the mechanism,
respectively.

An auction whose allocations lie in the $n$-simplex, i.e.\
$\vec{a}(\vec{v})\in\mathcal S_n$ for all $\vec{v}\in\vec{V}$, will be called a
\emph{lottery}, since its fractional allocations $a_i\in[0,1]$ can be
equivalently interpreted as the probability of assigning $1$ unit of service to
bidder $i$, given the linearity of the
utilities~\eqref{eq:utility_ex-post_definition}. In particular, lotteries with
only integral $0$-$1$ allocations, i.e.\ $\vec{a}\in\mathcal S_n\inters
\ssets{0,1}^n$ will be called \emph{deterministic auctions}. More generally, any
auction with allocation rule $\vec{a}\in\N^n$ will be called
\emph{integral}\label{page:integral-auction}.

\subsection{Incentive Compatibility}
\label{sec:truthfulness}

From the perspective of each bidder $i$, the goal is to bid so that they can
maximize their own utility. In particular, this means that bidders can lie and
misreport $b_i\neq v_i$. Therefore, one of the goals of mechanism design is to
construct auctions that avoid this pitfall, and which \emph{provably} guarantee
that truthful participation is to each bidder's best interest. From a
game-theoretic perspective, this can be formalized by demanding that truthful
bidding $b_i=v_i$ is an equilibrium of the induced Bayesian game.

This gives rise to the following constraints, known as \emph{dominant-strategy
incentive compatibility (DSIC)}: for any bidder $i$, any true value $v_i\in
V_i$, and any bidding profile $\vec{b}\in \vec{V}$, it holds that
\begin{equation}
    \label{eq:DSIC-def}
u_i(v_i,\vec{b}_{-i};v_i) \geq u_i(b_i,\vec{b}_{-i};v_i),
    \tag{\text{DSIC}}
\end{equation}
and its more relaxed version of \emph{Bayesian incentive compatibility (BIC)},
involving the interim utilities:
\begin{equation}
    \label{eq:BIC-def}
U_i(v_i;v_i) \geq U_i(b_i;v_i),
    \tag{\text{BIC}}
\end{equation}
for any bidder $i$, true value $v_i\in V_i$ and bid $b_i\in V_i$.

\paragraph{Individual rationality} Another desired property of our mechanisms is
that no bidder should harm themselves by truthfully participating in our
auction, known as \emph{individual rationality (IR)}. Similarly to the
truthfulness conditions \eqref{eq:DSIC-def} and \eqref{eq:BIC-def}, this can be
formalized both in an ex-post and interim way: $u_i(v_i,\vec{b}_{-i})\geq 0$ and
$U_i(v_i;v_i)\geq 0$, respectively, for all bidders $i$, true values $v_i\in
V_i$ and other bidders' bid profile $\vec{b}_{-i}\in \vec{V}_{-i}$,
respectively.

One elegant way to merge the (IR) constraints into truthfulness, is to extend
the bidding space of bidder $i$ in \eqref{eq:DSIC-def} and \eqref{eq:BIC-def}
from $V_i$ to $\bar{V}_i\coloneqq V_i\union\ssets{\noparticipation}$ and define 
\begin{equation} 
    \label{eq:no-participation-border}
    a_i(\noparticipation,\vec{b}_{-i})=p_i(\noparticipation,\vec{b}_{-i})=0
\end{equation}
for all bidders $i$ and other bidders' bids $\vec{b}_{-i}\in \vec{V}_{-i}$.
Then, bidding $\noparticipation$ can be interpreted as an option to ``abstain''
from the auction for a utility of
$u_i(\noparticipation,\vec{b}_{-i};v_i)=U_i(\noparticipation;v_i)=0$. From now
on we will assume that our truthfulness conditions \eqref{eq:DSIC-def} and
\eqref{eq:BIC-def} are indeed extended in that way to $\bar{V_i}$, thus
including the (IR) constraints. An auction will be called DSIC (resp.\ BIC) if
it satisfies those (extended) \eqref{eq:DSIC-def} (resp.\ \eqref{eq:BIC-def})
constraints. Observe that, since
$\eqref{eq:DSIC-def}\subseteq\eqref{eq:BIC-def}$, any DSIC auction is also BIC.

\paragraph{Optimal auctions} 
The main focus of our paper is the design of \emph{optimal auctions}, for
discrete value domains. That is, maximize the seller's \emph{revenue} within the
space of all feasible \emph{truthful} auctions. Formally, if for a given auction
$M=(\vec{a},\vec{p})$ we denote its expected revenue, with respect to the value
priors $\vecc F$, by 
\begin{equation}
\label{eq:revenue-def}
    \revenue(M)\coloneqq \expect[\vec v\sim\vec{F}]{\sum_{i=1}^n p_i(\vecc v)},
\end{equation}
then our optimization problem can be stated as $\sup_{M: \mathcal{A} \land
\eqref{eq:DSIC-def}} \revenue(M)$, or $\sup_{ M: \mathcal{A} \land
\eqref{eq:BIC-def}} \revenue(M)$, depending on whether we choose the notion
dominant-strategy, or Bayesian truthfulness. An optimal solution to the former
problem will be called \emph{optimal DSIC} auction, and to the latter,
\emph{optimal BIC} auction. Following the standard convention in the field (see,
e.g., \textcite{Krishna2009a} and \textcite{Myerson1981a}), the term
\emph{optimal auction} that does not explicitly specify the underlying
truthfulness notion, will refer to the optimal BIC auction. 
Notice that, since $\eqref{eq:DSIC-def}\subseteq\eqref{eq:BIC-def}$, for an
optimal DSIC auction $M$ and an optimal BIC auction $M'$ it must be that
$\revenue(M) \leq \revenue(M')$.

Nevertheless, as we demonstrate in~\cref{sec:KKT}, our general duality approach
provides for greater flexibility with respect to the optimization objective. For
example, this will allow us to instantiate our framework for a linear
combination of revenue and another important objective in auction theory, that
of \emph{social welfare}:
\begin{equation}
\label{eq:welfare-def}
    \welfare(M)\coloneqq \expect[\vec v\sim\vec{F}]{\sum_{i=1}^n a_i(\vecc v)v_i}.
\end{equation}

\subsection{Locality of Truthfulness}
\label{sec:IC-locality}
It turns out our truthfulness constraints can be simplified, and expressed
through a set of constraints that are ``local'' in nature, in the sense that
they only involve deviations between adjacent values. To formalize this, recall
that our value spaces $V_i$ are finite, so we can define the notion of
\emph{predecessor} and \emph{successor} values for a given bidder $i$ and a value
$v_i\in V_i$:
$$
v_i^+ \coloneqq \min\sset{v\in V_i\fwh{v>v_i}}
\quad\text{and}\quad
v_i^- \coloneqq \max\sset{v\in V_i\fwh{v<v_i}},
$$
if the above sets are non-empty, otherwise we define $v_i^+\coloneqq
\noparticipation$ for $v_i=\max V_i$ and $v_i^-\coloneqq \noparticipation$ for
$v_i=\min V_i$.

Now we can state the local characterization of truthfulness, first for
\eqref{eq:DSIC-def}, but a totally analogous lemma holds for \eqref{eq:BIC-def}
as well -- see~\cref{append:local-IC}. This result can be essentially derived
by the work of~\parencite[Theorems~1 and~2]{Malakhov2004}; for reasons of clarity and
compatibility with our model and notation, we also present a proof
in~\cref{append:local-IC}.

\begin{lemma}[{\textcite{Malakhov2004}}]
\label{lemma:local-DSIC}
	For any discrete, single-dimensional auction $(\vec a,\vec p)$, the
    \eqref{eq:DSIC-def} condition is \emph{equivalent} to the following set of
    constraints:
    \begin{align}
         u_i(\vec v;v_i) &\geq u_i(v_i^-,\vec{v}_{-i};v_i) \label{eq:local-IC-down}\\
	     u_i(\vec v;v_i) &\geq u_i(v_i^+,\vec{v}_{-i};v_i)\label{eq:local-IC-up},
    \end{align}
    for all bidders $i\in[n]$ and any value profile $\vec v\in\vec{V}$. 
    Furthermore, conditions \eqref{eq:local-IC-down} and \eqref{eq:local-IC-up} imply
     \begin{equation}
         a_i(\vec{v}) \geq a_i(v_i^{-},\vec{v}_{-i}) \label{eq:local-IC-monotonicity},
     \end{equation}
     for all $i\in[n]$ and $\vec v\in\vec{V}$.
\end{lemma}
Conditions \eqref{eq:local-IC-down} and \eqref{eq:local-IC-up} are called
\emph{downwards} and \emph{upwards} DSIC constraints, respectively, and
\eqref{eq:local-IC-monotonicity} are called \emph{monotonicity} constraints.

\section{The Discrete Myerson Auction: an LP Duality Approach}
\label{sec:Myerson}

In this section we begin our study of optimal single-parameter auctions, by
considering the canonical single-item setting of~\textcite{Myerson1981a}, but
under discrete values. That is, the feasibility set for our allocations is the
simplex $\mathcal{S}_n$ (see~\cref{sec:model-notation}), giving rise to the
following \emph{feasibility constraints}:
\begin{equation}\label{eq:feasibility_unity-simplex}
    \sum_{i=1}^{n} a_i( \vec{v}) \leq 1,
    \qquad\text{for all}\;\; \vec{v} \in \vec{V}.
\end{equation}
Our results of this section are summarized in the following main theorem:

\begin{mdframed}[backgroundcolor=gray!30]
    \vspace{-8pt}
\begin{theorem}[Optimal Discrete Single-Item Auction]
    \label{th:discrete-optimal-myerson}
    For any discrete, single-item auction setting, the following hold for revenue maximization:
    \begin{enumerate}
        \item\label{item:optimal-single_item-determinism} There always exists an
        optimal auction which is deterministic.
        \item \label{item:optimal-single_item-BICvsDIC} \emph{Any} optimal DSIC
        auction is an optimal BIC auction.
        \item\label{item:optimal-myerson-discrete-explicit} The following
        deterministic DSIC auction is optimal (even within the class of
        randomized BIC auctions):
        \begin{itemize}
        \item Allocate (fully) the item to the bidder with the highest
        non-negative ironed virtual value~\eqref{eq:ironed_virtual_values},
        breaking ties arbitrarily.\footnotemark
        \item Collect from the winning bidder a payment equal to their critical
        bid~\eqref{eq:first_payment_rule}.
        \end{itemize}
    \end{enumerate}
\end{theorem}
\end{mdframed}
\footnotetext{In order to maintain determinism, this can be any fixed
deterministic tie-breaking rule; e.g., allocating the bidder with the smallest
index $i$. Fractionally splitting the item among bidders that tie would still
ensure revenue optimality (and DSIC), but the mechanism would be
randomized.}

Point~\ref{item:optimal-myerson-discrete-explicit} of
\cref{th:discrete-optimal-myerson} is essentially a discrete analogue of
Myerson's optimal auction for the continuous case. As we mentioned in our
introduction (see~\cref{sec:related-work,sec:results}), this result can be
already derived by readily combining prior work on discrete auctions (see, e.g.,
\parencite{Cai2012b,Elkind2007}); our contribution here is not the result
itself, but the proof technique, which makes use of classical LP duality theory.
This allows us to make use of powerful and transparent results from polyhedral
combinatorics, to structurally characterize optimal auctions. In particular, we
establish the optimality of DISC mechanisms, in a very general sense
(see Point~\ref{item:optimal-single_item-BICvsDIC}), which to the best of our
knowledge was not known before. This is also enabled by our discrete
optimization view of the problem, through the use of polyhedral properties
(see~\cref{sec:DSIC-vs-BIC-Myerson}). Finally, observe that
Point~\ref{item:optimal-single_item-determinism} can be derived directly as a
corollary of Point~\ref{item:optimal-myerson-discrete-explicit}; nevertheless,
we choose to state it independently, in order to reflect the logical progression
of our derivation in this paper, which actually allows us to establish
Point~\ref{item:optimal-single_item-determinism} more generally, as a result of
the polyhedral structure of our problem (see
\cref{sec:determinism-vs-randomization-Myerson}), \emph{before} we determine the
actual optimal solution in Point~\ref{item:optimal-myerson-discrete-explicit}. 

We start our presentation by considering the revenue-maximization problem under
the more restricted DSIC truthfulness notion. We do this for reasons of clarity
of exposition, and then in~\cref{sec:DSIC-vs-BIC-Myerson} we carefully discuss
how our formulations adapt for the more relaxed \eqref{eq:BIC-def} constraints,
and the relation between the two notions with respect to optimality, completing
the picture for~\cref{th:discrete-optimal-myerson}.

\subsection{A Chain of Dual Linear Programs}
\label{sec:duality-chain-myerson}

In this section we develop the skeleton of our approach for
proving~\cref{th:discrete-optimal-myerson}. It consists of a sequence of LPs, as
summarized in~\cref{fig:chain_tikz}. We start by formulating the single-item,
revenue-maximization problem as an LP in~\eqref{LP:LP1}. Next, we dualize it
in~\eqref{LP:DP1}, and then restrict the program to derive~\eqref{LP:DP2} that
can only have a worse (i.e., higher) optimal objective. Then, we dualize again,
deriving a maximization program in~\eqref{LP:LP2}. Finally, we prove
(see~\cref{lemma:DSIC_opt}) that our original maximization
program~\eqref{LP:LP1} is a relaxation of~\eqref{LP:LP2}, thus establishing a
collapse of the entire duality chain, and the equivalence of all involved LPs.
This closure of the chain is exactly from where virtual
values~\eqref{eq:virtual_values-def}, virtual welfare
maximization~\eqref{LP:LP2}, optimality of determinism
(see~\cref{lemma:determinism_TU}), and the optimal payment rule~\eqref{LP:LP2}
naturally emerge.

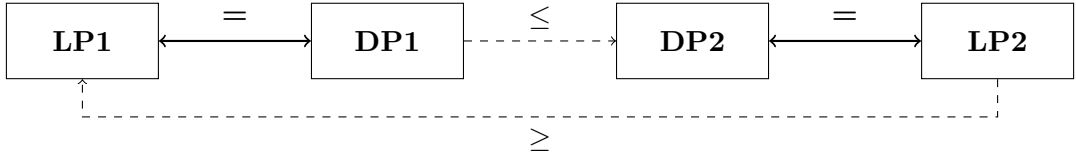
\begin{figure}
    \centering
    	\begin{tikzpicture}[ node distance=1.9cm and 1.9cm, every node/.style={draw,
			minimum width=2cm, minimum height=1cm, align=center, font=\bfseries}
			]
			
			\node (rect1) {LP1};
			\node[right=of rect1] (rect2) {DP1};
			\node[right=of rect2] (rect3) {DP2};
			\node[right=of rect3] (rect4) {LP2};
			
			\draw[->, dashed]
			(rect4.south) -- ++(0,-.5) coordinate (midpoint1)
			-- (midpoint1 -| rect1.south) coordinate (midpoint2)
			node[midway, below=-0.2, draw=none] {$\geq$} -- (rect1.south);
			\draw[<->, thick]
				(rect1) -- node[midway, above=-0.2, draw=none] {=} (rect2);
			\draw[->, dashed]
				(rect2) -- node[midway, above=-0.2, draw=none] {$\leq$} (rect3);
			\draw[<->, thick]
				(rect3) -- node[midway, above=-0.2, draw=none] {=} (rect4);
		\end{tikzpicture}
    \Description{A chain of dual programs that closes ensures equality of all objective values}
    \caption{An overview of the linear programs used in our derivation
    throughout~\cref{sec:Myerson} and the relations between their optimal
    values.}
    \label{fig:chain_tikz}
\end{figure}

Before we formally present and start working within the LPs, we need to fix
some notation.
\paragraph{LP notation} Since our value sets are finite, for each player $i$ we
can enumerate their support as $V_i=\ssets{v_{i,1},v_{i,2},\dots,v_{i,K_i}}$,
for some positive integer $K_i$. For notational convenience, we denote $\vec{K}
\coloneqq [K_1]\times [K_2] \times\dots\times [K_n]$ and $\vec{K}_{-i} \coloneqq
[K_1]\times \dots \times [K_{i-1}] \times [K_{i-1}] \times \dots\times [K_n]$.
To keep our LP formulations below as clean as possible, we will feel free to
abuse notation and use the support indices $k \in [K_i]$ instead of the actual
values $v_{i,k} \in V_i$, as arguments for the allocations $a_i$, payments
$p_i$, and prior cdf's $F_i$ and pmf's $f_i$. That is, e.g., we will denote
$a_i(k,\vec{k}_{-i})$, $p_i(k,\vec{k}_{-i})$, $f_i(k)$, and $F_i(k)$, instead of
$a_i(v_i,\vec{v}_{-i})$, $p_i(v_i,\vec{v}_{-i})$, $f_i(v_i)$, and $F_i(v_i)$,
respectively, when the valuation profile $\vec{v}$ is such that $v_{i} =
v_{i,k_i}$ for $i\in[n]$. As all values are independently drawn from
distributions $F_i$, the probability of a bid profile $\vec{v} \in \vec{V}$
being realized is given by the pmf of their product distribution $\vec{F}$,
denoted by $f(\vec{k})=f(\vec{v}) = \prod_{i\in[n]} f_i(v_{i,k})$. Analogously,
we denote $f(\vec{k}_{-i}) = f(\vec{v}_{-i}) = \prod_{j\in[n]\setminus\ssets{i}}
f_j(v_{j,k_j})$. Finally, given that we make heavy use of duality, we choose to
label each constraint of our LPs with the name of its corresponding dual
variable, using \textcolor{blue}{blue colour} (see, e.g., \eqref{LP:LP1}).

\smallskip

For our starting~\eqref{LP:LP1}, we want to formulate an LP maximizing expected
revenue~\eqref{eq:revenue-def}, under the single-item allocation
constraints~\eqref{eq:feasibility_unity-simplex} of our current section, and
\ref{eq:DSIC-def} truthfulness, through its equivalent formulation
via~\cref{lemma:local-DSIC}. Since we want to optimize over the space of all
feasible auctions, the real-valued variables of our LP are the allocation and
payment rules of the auction, over all possible bidding profiles, namely
$\ssets{a_i(\vec{v}),p_i(\vec{v})}_{\vec{v}\in\vec{V}}$. Putting everything
together, we derive the following LP:
\begin{align*}
	 \text{max}& \quad  \sum_{\vec{v} \in \vec{V}} \sum_{i=1}^{n} p_i( \vec{v}) f(\vec{v}) \tag{\textbf{LP1}}\label{LP:LP1} \\  
    \text{s.t.} \quad &  v_{i,k} a_i(k, \vec{k}_{-i}) - p_i(k, \vec{k}_{-i}) \geq v_{i,k} a_i(k-1, \vec{k}_{-i}) - p_i(k-1, \vec{k}_{-i}) ,&& \textcolor{blue}{[\lambda_i(k,k-1, \vec{k}_{-i})]}\\
	&\quad \text{for} \; i \in [n], k \in [K_i], \vec{k}_{-i} \in \vec{K}_{-i}, \\
	&  v_{i,k} a_i(k, \vec{k}_{-i}) - p_i(k, \vec{k}_{-i}) \geq v_{i,k} a_i(k+1, \vec{k}_{-i}) - p_i(k+1, \vec{k}_{-i}) ,&& \textcolor{blue}{[\lambda_i(k,k+1, \vec{k}_{-i})]}\\
	&\quad \text{for} \; i \in [n], k \in [K_i], \vec{k}_{-i} \in \vec{K}_{-i}, \\
	& a_i(k, \vec{k}_{-i}) \geq a_i(k-1, \vec{k}_{-i}), && \textcolor{blue}{[\tau_i(k,k-1, \vec{k}_{-i})]}\\
	&\quad \text{for} \; i \in [n], k \in [K_i], \vec{k}_{-i} \in \vec{K}_{-i}, \\
	&\sum_{i=1}^{n} a_i( \vec{v}) \leq 1, && \textcolor{blue}{[\psi(\vec{v})]} \\
	&\quad \text{for} \; \vec{v} \in \vec{V}.
\end{align*}
Notice how our LP can readily incorporate the no-participation IR
constraints~\eqref{eq:no-participation-border}, by fixing the
under-/overflowing corner cases as constants 
\begin{equation}
    \label{eq:borderline-LP1}
    a_i(0,\vec{k}_{-i})=p_i(0,\vec{k}_{-i})=a_i(K_i+1,\vec{k}_{-i})=p_i(K_i+1,\vec{k}_{-i})=0
\end{equation}
for all bidders $i$, on any bidding profile $\vec{k}_{-i}$ of the other bidders.

According to this we formulate the dual LP \eqref{LP:DP1}. Similar to the
borderline cases~\eqref{eq:borderline-LP1} in the primal LP some restrictions on
the dual variables are necessary to obtain a correct dual problem formulation.
There we have 
\begin{equation}\label{eq:borderline-DP1}
    \lambda_i(K_i,K_i+1, \vec{k}_{-i}) = \lambda_i(K_i+1,K_i,
\vec{k}_{-i}) = \lambda_i(0,1,\vec{k}_{-i}) = \tau_i(K_i+1,K_i, \vec{k}_{-i}) = 0
\end{equation}
for all bidders $i$, on any bidding profile $\vec{k}_{-i}$ of the other bidders,
for constraints that do not exist in~\eqref{LP:LP1}. To ensure dual feasibility,
all dual variables corresponding to inequality constraints in the primal have to
be non-negative, thus all $\lambda,\psi,\tau \geq 0$. It is worth pointing out
that $\lambda_i(1,0,\vec{k}_{-i})$ and $\tau_i(1,0,\vec{k}_{-i})$ are explicitly
not fixed to zero as the corresponding constraints, the local downward DSIC
constraint that ensures IR, $v_{i,1} a_i(1, \vec{k}_{-i}) - p_i(1, \vec{k}_{-i})
\geq 0$, as well as the monotonicity constraint that ensures non-negativity of
the allocation variables, $a_i(1, \vec{k}_{-i}) \geq 0$, are crucial for the
problem. By that we write the dual LP as

\begin{align*}
	\min \quad & \sum_{\vec{v} \in \vec{V}} \psi(\vec{v}) \tag{\textbf{DP1}}\label{LP:DP1}\\
	\text{s.t.} \quad& \psi(k,\vec{k}_{-i}) \geq  v_{i,k} \lambda_i(k,k-1,\vec{k}_{-i}) + v_{i,k} \lambda_i(k,k+1,\vec{k}_{-i}) \\ & \qquad -  v_{i,k+1} \lambda_i(k+1,k,\vec{k}_{-i}) - v_{i,k-1} \lambda_i(k-1,k,\vec{k}_{-i})\\ & \qquad  + \tau_i(k,k-1,\vec{k}_{-i}) - \tau_i(k+1,k,\vec{k}_{-i}), && \textcolor{blue}{ [a_i(k,\vec{k}_{-i})]}\\
	&\quad \text{for \;} i \in [n], k \in [K_i],\vec{k}_{-i} \in \vec{K}_{-i}, \\
	& \lambda_i(k,k-1,\vec{k}_{-i}) + \lambda_i(k,k+1,\vec{k}_{-i}) \\ & \qquad -   \lambda_i(k+1,k,\vec{k}_{-i}) -  \lambda_i(k-1,k,\vec{k}_{-i}) = f(\vec{v}), &&\textcolor{blue}{[p_i(k,\vec{k}_{-i})]} \\
	&\quad \text{for \;} i \in [n], k \in [K_i],\vec{k}_{-i} \in \vec{K}_{-i}.
\end{align*}	

As $p$ are free variables in the primal problem, the corresponding dual
constraints are equations, while as $a$ are required to be non-negative, the
corresponding dual constraints are inequalities. In the same spirit as denoting
the local DSIC constraints, that consider a deviation to the lower value, as
\emph{downwards constraint}~\eqref{eq:local-IC-down}, we call the corresponding
dual variables $\lambda_i(k,k-1,\vec{k}_{-i})$ where the index in the first
argument is greater than in the second \emph{downward $\lambda$ variables}. The
dual variables $\lambda_i(k,k+1,\vec{k}_{-i})$ corresponding to the upwards DSIC
constraints~\eqref{eq:local-IC-up} are the \emph{upward $\lambda$ variables}.
Putting together the dual borderline variables~\eqref{eq:borderline-DP1} and the
set of equations in~\eqref{LP:DP1} we can state the following lemma, whose proof
can be found in~\cref{append:lemma:lambda_positive-proof}.

\begin{lemma}\label{lemma:lambda_positive} In any feasible solution of
    \eqref{LP:DP1} all downward $\lambda$ variables are strictly positive, i.e.,
    $$\lambda_i(k,k-1,\vec{k}_{-i}) > 0,$$
    for all $i \in [n], k\in [K_i], \vec{k}_{-i} \in \vec{K}_{-i}.$
\end{lemma}

This motivates us to reformulate the dual program in a certain way. Recall, that
any dual solution has to satisfy the set of equations
\begin{equation}\label{eq:dual_equations}
    \lambda_i(k,k-1,\vec{k}_{-i}) + \lambda_i(k,k+1,\vec{k}_{-i}) = f(\vec{v}) + \lambda_i(k+1,k,\vec{k}_{-i}) +  \lambda_i(k-1,k,\vec{k}_{-i}).
\end{equation}
Using this we reformulate the dual inequality constraints
\begin{align*}
    \psi(\vec{v}) &\stackrel{\eqref{eq:dual_equations}}{\geq}  v_{i,k} f(\vec{v}) - ( v_{i,k+1} - v_{i,k}) \lambda_i(k+1,k,\vec{k}_{-i}) + (v_{i,k} - v_{i,k-1}) \lambda_i(k-1,k,\vec{k}_{-i})\\
    & \qquad  + \tau_i(k,k-1,\vec{k}_{-i}) - \tau_i(k+1,k,\vec{k}_{-i})
\end{align*}
Note, that by the use of~\eqref{eq:dual_equations}, i.e. exclusively equations,
this is only a reformulation and does not affect the set of feasible dual
solutions of \eqref{LP:DP1}. Now we unconventionally fix specific values of the
$\lambda$ variables. As the dual's objective aims to minimize the sum of the
$\psi$ variables, according to the reformulated inequality constraints it seems
convenient to choose all upward $\lambda$ as small and all downward $\lambda$ as
large as possible. To do so we set $\lambda_i(k,k+1,\vec{k}_{-i}) = 0$, for all
$k \in [K_i], i \in [n]$ and $\vec{k}_{-i}\in \vec{K}_{-i}$. Fixing variables,
essentially adding equality constraints, can only increase the optimal value of
\eqref{LP:DP1} in terms of minimization. As a next critical step, we introduce
\emph{free} variables $\rho$ and substitute the expression
$$\rho_i(k,\vec{k}_{-i})\coloneqq\lambda_i(k,k-1,\vec{k}_{-i}) -
\lambda_i(k+1,k,\vec{k}_{-i})$$ for all bidders $i$ with value index
$k\in[K_i]$, and any bidding profile $\vec{k}_{-i}$ of the other bidders. These
variables are all bound to fixed values and by dropping the $\lambda$ variables
from the problem formulation we do not lose any information about feasible dual
solutions as by ${\lambda_i(K+1,K,\vec{k}_{-i})=0}$ we keep track of all fixed
values. The reformulated dual LP then is

\begin{align*}
	\min \quad & \sum_{\vec{v} \in \vec{V}} \psi(\vec{v}) \tag{\textbf{DP2}}\label{LP:DP2} \\
	\text{s.t.} \quad& \psi(k,\vec{k}_{-i}) \geq  v_{i,k} \rho_i(k,\vec{k}_{-i}) - ( v_{i,k+1} - v_{i,k}) \sum_{l=k+1}^{K_i} \rho_i(l,\vec{k}_{-i}) \\ & \qquad  + \tau_i(k,k-1,\vec{k}_{-i}) - \tau_i(k+1,k,\vec{k}_{-i}), && \textcolor{blue}{ [a_i(k,\vec{k}_{-i})]}\\
	&\quad \text{for \;} i \in [n], k \in [K_i],\vec{k}_{-i} \in \vec{K}_{-i}, \\
	& \rho_i(k,\vec{k}_{-i}) = f(\vec{v}) , &&\textcolor{blue}{[p_i(k,\vec{k}_{-i})]}\\
	&\quad \text{for \;} i \in [n], k \in [K_i],\vec{k}_{-i} \in \vec{K}_{-i}. 
\end{align*}
The inequality constraints now can also be written with all explicit values of
$\rho$ inserted. By that we obtain for a fixed bidder $i$ and bids
$\vec{v}_{-i}$
\begin{equation}\label{eq:virtual_values_appear}
    \psi(k,\vec{k}_{-i}) \geq  f(\vec{v}) \left[ v_{i,k} - ( v_{i,k+1} - v_{i,k})  \frac{1-F_i(k)}{f_i(k)}  + \frac{\tau_i(k,k-1,\vec{k}_{-i})}{f(\vec{v})} - \frac{\tau_i(k+1,k,\vec{k}_{-i})}{f(\vec{v})} \right].
\end{equation}

This gives rise to the well known definition of a sequence of values for player
$i$ which is independent of all other bidders' values $\vec{v}_{-i}$.

\begin{definition}[Virtual Values]\label{definition:virtual_values}
    The \emph{virtual values} of bidder $i \in [n]$ are defined as
    \begin{equation}
        \label{eq:virtual_values-def}
        \varphi_i(k) = \varphi_i(v_{i,k}) \coloneqq v_{i,k} - ( v_{i,k+1} - v_{i,k}) \frac{1-F_i(v_{i,k})}{f_i(v_{i,k})} \qquad\text{for}\;\;k\in [K_i].
    \end{equation}
\end{definition}
We return to the primal setting of allocation and payment variables by now
taking \emph{the dual of the dual}. To get the full transparency of the gained
insights within the reformulation to~\eqref{LP:DP2} we do two things at the same
time: We insert the true values of all $\rho$ in the inequalities and obtain the
virtual values as the coefficients of the allocation variables in the new primal
objective. At the same time, we stick with $\rho$ as free variables in the dual
inequalities and obtain the payment formula in~\eqref{LP:LP2} as the
coefficients of $\rho$ in the dual become the coefficients of the allocation
variables in the primal payment formula. Note that equivalently, we could still
maximize the expected payments in the new primal LP without using the explicit
values for $\rho$.

\begin{align*}
	\max \quad &   \sum_{\vec{v} \in \vec{V}} \sum_{i=1}^{n} a_i( \vec{v}) \varphi_i(k) f(\vec{v}) \tag{\textbf{LP2}}\label{LP:LP2} \\ 
	\text{s.t.} \quad&  p_i(k, \vec{k}_{-i}) = v_{i,k} a_i(k, \vec{k}_{-i}) - \sum_{l=1}^{k-1} (v_{i,l+1} - v_{i,l} ) a_i(l, \vec{k}_{-i}),&& \textcolor{blue}{[\rho_i(k, \vec{k}_{-i})]}\\
	&\quad \text{for} \; i \in [n], k \in [K_i], \vec{k}_{-i} \in \vec{K}_{-i}, \\
	& a_i(k, \vec{k}_{-i}) \geq a_i(k-1, \vec{k}_{-i}), && \textcolor{blue}{[\tau_i(k,k-1, \vec{k}_{-i})]}\\
	&\quad \text{for} \; i \in [n], k \in [K_i], \vec{k}_{-i} \in \vec{K}_{-i}, \\
	&\sum_{i=1}^{n} a_i( \vec{v}) \leq 1, && \textcolor{blue}{[\psi(\vec{v})]} \\
	&\quad \text{for} \; \vec{v} \in \vec{V}.
\end{align*}

As our interest lies in optimal auctions, we \emph{close the chain} of LPs using
\cref{lemma:lambda_positive} and strong LP duality to verify that the sets of
optimal solutions of~\eqref{LP:LP1} and of~\eqref{LP:LP2} are equivalent.

\begin{lemma} \label{lemma:DSIC_opt} Any optimal solution of \eqref{LP:LP2}
	represents an optimal DSIC auction, i.e.\ an optimal solution of
	\eqref{LP:LP1} and vice versa.
\end{lemma}

\begin{proof}
    We first show that any solution of \eqref{LP:LP2} is also feasible for
    \eqref{LP:LP1}. This closes the chain in terms of objective values.
    Secondly, we use complementary slackness to prove that any optimal solution
    of \eqref{LP:LP1} has to be feasible for \eqref{LP:LP2}. By that, both linear
    programs have the exact same set of optimal solutions.

    Let $(\vec{a},\vec{p})$ be an optimal solution of \eqref{LP:LP2}. We insert
    the payment rule into the local DSIC constraints of \eqref{LP:LP1} to verify
    that they are satisfied. For that fix player $i$, player $i$'s value
    $v_{i,k}$, and all other players' values $\vec{v}_{-i}$ and consider the
    downward constraint

    \begin{align*}
				v_{i,k} a_i(k, \vec{k}_{-i}) - p_i(k, \vec{k}_{-i}) &\geq v_{i,k} a_i(k-1, \vec{k}_{-i}) - p_i(k-1, \vec{k}_{-i})\\
				 \sum_{l=1}^{k-1} (v_{i,l+1} - v_{i,l} ) a_i(l, \vec{k}_{-i}) &\geq v_{i,k} a_i(k-1, \vec{k}_{-i}) - v_{i,k-1} a_i(k-1, \vec{k}_{-i}) + \sum_{l=1}^{k-2} (v_{i,l+1} - v_{i,l} ) a_i(l, \vec{k}_{-i})\\
				(v_{i,k} - v_{i,k-1} ) a_i(k-1, \vec{k}_{-i}) &\geq (v_{i,k}  - v_{i,k-1} ) a_i(k-1, \vec{k}_{-i}).
		\end{align*}
  The equality in the last line is no coincidence, and we have a closer look at
  it in the last part of the proof.

   For the upward constraint we have
    \begin{align*}
				v_{i,k} a_i(k, \vec{k}_{-i}) - p_i(k, \vec{k}_{-i}) &\geq v_{i,k} a_i(k+1, \vec{k}_{-i}) - p_i(k+1, \vec{k}_{-i})\\
				\sum_{l=1}^{k-1} (v_{i,l+1} - v_{i,l} ) a_i(l, \vec{k}_{-i})  &\geq v_{i,k} a_i(k+1, \vec{k}_{-i}) - v_{i,k+1} a_i(k+1, \vec{k}_{-i}) + \sum_{l=1}^{k} (v_{i,l+1} - v_{i,l} ) a_i(l, \vec{k}_{-i})\\
				(v_{i,k+1} - v_{i,k}) a_i(k+1, \vec{k}_{-i}) &\geq (v_{i,k+1} - v_{i,k} ) a_i(k, \vec{k}_{-i})
			\end{align*}
   which always holds by monotonicity.

   Finally, for the other direction, let $(\vec{a},\vec{p})$ be an optimal
   solution of \eqref{LP:LP1}. Any optimal solution of this linear program by
   strong duality has to satisfy complementary slackness: If any dual variable
   is strictly positive, the corresponding primal constraint has to bind. By
   \cref{lemma:lambda_positive} we know that all feasible, thus, all optimal
   downward $\lambda$ variables, are positive. This implies that in any optimal
   solution of \eqref{LP:LP1} all local downward constraints have to bind. The
   payment rule of \eqref{LP:LP2} is only the result of the successive
   application of the binding constraints. The upward constraints are then also
   satisfied, as this follows directly from the first part of the proof.
\end{proof}

The immediate result is that the problem of finding an optimal \ref{eq:DSIC-def}
auction reduces to finding an optimal solution of~\eqref{LP:LP2}, i.e. a
feasible, virtual welfare maximizing, monotone allocation rule $\vec{a}$. The
optimal payments are computed afterwards as a linear function of the allocations
according to the payment rule 
\begin{equation}\label{eq:first_payment_rule}
    p_i(k, \vec{k}_{-i}) = v_{i,k} a_i(k, \vec{k}_{-i}) - \sum_{l=1}^{k-1} (v_{i,l+1} - v_{i,l} ) a_i(l, \vec{k}_{-i}).
\end{equation}

If for player $i$ and fixed $\vec{v}_{-i}$ the allocation variables $a_i(k,
\vec{k}_{-i})\in \{0,1\}$ for $k\in [K_i]$, the payment rule
\eqref{eq:first_payment_rule} simplifies: If player $i$ wins, the payment
$p_i(k, \vec{k}_{-i})$ is the \emph{critical bid}, i.e. the minimum value such
that player $i$ still wins, and zero if player $i$ does not win. Therefore, we
want to examine the potential of the allocation variables being binary in an
optimal solution in the following.

\subsection{Deterministic vs Randomized Auctions}
\label{sec:determinism-vs-randomization-Myerson}

In this section we essentially establish the foundation for
Point~\ref{item:optimal-single_item-determinism} of
\cref{th:discrete-optimal-myerson}. We are using the property of \emph{total
unimodularity} \cite{Hoffman2010} of the constraint matrix of \eqref{LP:LP2}.
This is enough to show that the optimal allocations of \eqref{LP:LP1} and
\eqref{LP:LP2} are the convex hull of optimal binary solutions.

\begin{lemma}[Optimality of determinism]\label{lemma:determinism_TU} The
    vertices of the polyhedron of feasible allocations for \eqref{LP:LP2} are
    integral, hence, binary.
\end{lemma}

\begin{proof}
	This proof is based on the matrix property of total unimodularity (TU). More
	specific, a matrix $A$ is called TU, if any quadratic submatrix has determinant $-1,0$ or $1$. We will make use of the following well-known properties (see,
	e.g., \cite[Ch.~13]{schrijver2003co}): for any matrix
	$A\in\ssets{-1,0,1}^{M\times N}$, (i) $A$ is TU if and only if
	$(A\,\unitmatrix)$ is TU, where $\unitmatrix$ is an $M\times M$ unit matrix, (ii) $A$ is TU if and only if $A^\top$ is TU, where $A^\top$ denotes the transpose of matrix $A$, and (iii) the adjacency matrix of a directed graph is TU.
	
	To argue this coherently we have to put some effort in understanding the
	structure of the linear constraints that can be expressed as matrix vector
	inequalities. The set of feasible allocations for \eqref{LP:LP2} are all
	non-negative $\vec{a}$ satisfying
	\begin{align}
		& a_i(k, \vec{k}_{-i}) \geq a_i(k-1, \vec{k}_{-i})
		&&\text{for} \; i \in [n], k \in [K_i], \vec{k}_{-i} \in \vec{K}_{-i}, \label{eq:monotonicity_lemma}\\
		&\sum_{i=1}^{n} a_i( \vec{v}) \leq 1 &&\text{for} \; \vec{v} \in \vec{V}.\label{eq:feasibility_lemma}
	\end{align}
	Therefore, we consider the allocation polyhedron $\{\vec{a} \, | \,
	a_i(k, \vec{k}_{-i}) \geq a_i(k-1, \vec{k}_{-i}) \, , i \in [n], k \in
	[K_i], \vec{k}_{-i} \in \vec{K}_{-i},\, \text{and}, \, \sum_{i=1}^{n} a_i(
	\vec{v}) \leq 1 \, ,\vec{v} \in \vec{V}\}$ as $\{\vec{a} \, | \, A \vec{a}
	\leq b \}$. We see the allocations $\vec{a}$ as a vector of dimension
	$\mathcal{N} \coloneqq n \cdot |\vec{V}| = n \cdot |V_1| \cdots |V_n| = n
	\cdot K_1 \cdots K_n$. We order them such that the first $|\vec{V}|$ entries
	are all allocations of player $1$ varying over the bid profiles $\vec{v}$,
	then the second player and so on. The order of the $\vec{v}$ is consistent
	over all players. By that we can write the constraint matrix $A$ and the
	right-hand side $b$ as
	\begin{equation}\label{eq:TU_matrix}
		A = \left(\begin{array}{ccc|c|ccc}
			M_1(\vec{v}_{-1}) &  & 0 &  \\
			& \ddots & & \zeromatrix & & \zeromatrix \\
			0 &  & M_1(\vec{v}^\prime_{-1}) &  \\\hline
			& \zeromatrix && \ddots & &  \zeromatrix \\\hline
			&&&& M_n(\vec{v}_{-n}) &  & 0 \\
			& \zeromatrix && \zeromatrix & & \ddots & \\
			&&&& 0 &  & M_n(\vec{v}^\prime_{-n}) \\\hline
			1 &  & 0 &  & 1 & &0 \\
			& \ddots &  & \cdots && \ddots  \\
			0 &  & 1 & & 0 && 1 
		\end{array}\right), \quad 
		b = \begin{pmatrix}
			0 \\
			\vdots \\
			0 \\
            0 \\
			0 \\
			\vdots \\
			0 \\
			1 \\
			\vdots \\
			1 
		\end{pmatrix}
	\end{equation}
	Here, $M_i(\vec{v}_{-i})$ is a $|V_i|\times|V_i|$ matrix that contains the
	coefficients of the monotonicity constraints~\eqref{eq:monotonicity_lemma}
	of player $i$ fixing all other players' values to $\vec{v}_{-i}$. We do not
	keep explicit track of the order within the monotonicity block matrices, but
	note that each block is the transposed adjacency matrix of a directed graph.
	By that, the monotonicity part is TU. Note, that the monotonicity constraints $a_i(1,
	\vec{k}_{-i}) \geq a_i(0, \vec{k}_{-i})=0$ are not considered in the
	monotonicity, as they can be seen as general non-negativity constraints. These constraints, however, are not decisive for $A$ being TU as they extend the constraint matrix only in the sense of $(A\,\unitmatrix)$. Therefore, we do not consider them.
	
	The side by side unity matrices in the last row of blocks represent the feasibility constraints~\eqref{eq:feasibility_lemma}
	and are also clearly TU. It remains to show that any sub-matrix of $A$ consisting of a mixture of the
	different parts of monotonicity blocks and the side by side unit
	matrices is TU. 
	
	For that, we consider the determinant of an arbitrary submatrix of $A$ that only contains elements of one row of the side by side unit matrices as the last row of this submatrix. If only elements corresponding to the columns of a single player are chosen, the problem reduces to (i) and (iii). Therefore, we assume to have entries corresponding to different players. We first consider the case where all elements in this row are $1$ entries and call this the \textit{simple case}. As each of those corresponds to a different player and we can only choose players$-1$ rows of the upper part for the submatrix, due to the block structure there has to exist a column with a single $1$ entry in the last row and only $0$ entries above. The Laplace expansion along this column clearly leads to the desired result.
	
	Next, we add an arbitrary $0$ element to the last row of the submatrix, i.e., we add a column above this zero and we add a row of elements from the upper block part. Here, we assume that the submatrix did not have a row of zero entries before adding the new elements, otherwise the Laplace expansion along this row immediately would reduce to the simple case. Now only three cases can occur: (a) the new $0$ element does not correspond to any player for which a $1$ entry in the last row exists, (b) the new $0$ element is added to a player for which a $1$ entry exists in the last row and the column above contains two non-zero entries, or (c) the new $0$ element is added to a player for which a $1$ entry exists in the last row and the column above contains exactly one non-zero entry. If (a) applies, there can be at most one non-zero entry exactly where the new row is added. The Laplace expansion along this new column reduces to the simple case, the same holds for case (c). In the case of (b) the new row can only contain at most a single non-zero entry. The Laplace expansion along this new row reduces to the simple case again.
	
	Lastly, a submatrix with multiple rows from the side-by-side unit matrices easily reduces to the case of a single row of this lower part: As this part never has more than a single $1$ entry in each column the Laplace expansion along any column reduces to a single row again. 
	
	Therefore, $A$ is totally unimodular and as $b$ has only integer entries, the proof is concluded.
\end{proof}

Thus, determinism of optimal DSIC auctions is without loss. Also, since the set
of optimal solutions is convex, any fractional optimal solution is only a convex
combination of multiple integer solutions and for given $\vec{v} \in \vec{V}$
represents a probability distribution.

\subsection{Dominant-strategy vs Bayesian Truthfulness}
\label{sec:DSIC-vs-BIC-Myerson}

The optimal auction problem typically is considered in a setting where
truthfulness constraints are a relaxed version of \eqref{eq:DSIC-def} and a
bidder's truthfulness only has to hold in expectation over all other bidders'
distributions, i.e., in the \eqref{eq:BIC-def} sense
(see~\cref{sec:truthfulness}). In this section we essentially perform the same
steps as in \cref{sec:duality-chain-myerson} where we considered DSIC
truthfulness, but unlike before under the constraints that an auction has to be
BIC. Now, by the expectations in the Bayesian setting, we have a drastically
reduced number of constraints. Due to the feasibility constraints
\eqref{eq:feasibility_unity-simplex} that do not change between the two
settings, we still consider the same number of primal variables. Maintaining the
amount of dual constraints ultimately yields the same payment formula as
in~\eqref{LP:LP2} for any optimal BIC auction.

The LP to find the optimal auction under Bayesian truthfulness is
\eqref{LP:BLP1}. \cref{lemma:local-DSIC} allows us to restrict ourselves to
local truthfulness without loss as well. Also note, that we fix the same
borderline variables to the same values as in the DSIC setting,
see~\eqref{eq:borderline-LP1}.

\begin{align*}
        \max \quad &   \sum_{\vec{v} \in \vec{V}} \sum_{i=1}^{n} p_i( \vec{v}) f(\vec{v}) \tag{\textbf{BLP1}}\label{LP:BLP1} \\ 
		\text{s.t.} \quad&  \sum_{\vec{v}_{-i}} \bigl[ v_{i,k} a_i(k, \vec{k}_{-i}) - p_i(k, \vec{k}_{-i}) - v_{i,k} a_i(k-1, \vec{k}_{-i}) + p_i(k-1, \vec{k}_{-i}) \bigr] f_{-i}(\vec{v}_{-i}) \geq 0  , && \textcolor{blue}{[\lambda_i(k,k-1)]}\\
        &\quad \text{for} \; i \in [n], k \in [K_i], \\
        &  \sum_{\vec{v}_{-i}} \bigl[ v_{i,k} a_i(k, \vec{k}_{-i}) - p_i(k, \vec{k}_{-i}) - v_{i,k} a_i(k+1, \vec{k}_{-i}) + p_i(k+1, \vec{k}_{-i}) \bigr] f_{-i}(\vec{v}_{-i}) \geq 0  ,&& \textcolor{blue}{[\lambda_i(k,k+1)]}\\
		&\quad \text{for} \; i \in [n], k \in [K_i],  \\
		&\sum_{\vec{v}_{-i}} \bigl[  a_i(k, \vec{k}_{-i}) - a_i(k-1, \vec{k}_{-i})  \bigr] f_{-i}(\vec{v}_{-i}) \geq 0  , && \textcolor{blue}{[\tau_i(k,k-1)]}\\
		&\quad \text{for} \; i \in [n], k \in [K_i],  \\
		&\sum_{i=1}^{n} a_i( \vec{v}) \leq 1, && \textcolor{blue}{[\psi(\vec{v})]} \\
		&\quad \text{for} \; \vec{v} \in \vec{V}.
	\end{align*}
The dual LP \eqref{LP:BDP1} now has a reduced number of variables but the same
number of constraints as \eqref{LP:DP1}.

\begin{align*}
		\min \quad & \sum_{\vec{v} \in \vec{V}} \psi(\vec{v}) \tag{\textbf{BDP1}}\label{LP:BDP1}\\
		\text{s.t.} \quad& \psi(\vec{v}) \geq f(\vec{v}_{-i}) \bigl[ v_{i,k} \lambda_i(k,k-1) + v_{i,k} \lambda_i(k,k+1) \\ & \qquad -  v_{i,k+1} \lambda_i(k+1,k) - v_{i,k-1} \lambda_i(k-1,k)\\ & \qquad  + \tau_i(k,k-1) - \tau_i(k+1,k) \bigr], && \textcolor{blue}{[a_i(k,\vec{k}_{-i})]}\\
		&\quad \text{for} \; i \in [n], k \in [K_i], \vec{k}_{-i} \in \vec{K}_{-i} ,\\
		&  \lambda_i(k,k-1) + \lambda_i(k,k+1) \\ & \qquad -   \lambda_i(k+1,k) -  \lambda_i(k-1,k) =   f_i(v_{i,k}),  && \textcolor{blue}{[p_i(k,\vec{k}_{-i})]} \\
		&\quad \text{for} \; i \in [n], k \in [K_i], \vec{k}_{-i} \in \vec{K}_{-i}.
\end{align*}

By the same argument as in the DSIC setting, \cref{lemma:lambda_positive} still
applies, and all feasible downward $\lambda$ are strictly positive. Again
we fix all upward $\lambda$ to zero and after the similar substitution by the
free variable $\rho$, we obtain the same results for the dual constraints,
namely,

\begin{align*}
    &\psi(\vec{v}) \geq f(\vec{v}_{-i}) \Bigg[ v_{i,k} \rho_i(k, \vec{k}_{-i}) - (v_{i,k+1} - v_{i,k}) \sum_{l=k+1}^{K_i} \rho_i(l, \vec{k}_{-i})  + \tau_i(k,k-1) - \tau_i(k+1,k) \Bigg], \\
    &\rho_i(k, \vec{k}_{-i}) \coloneqq \lambda_i(k,k-1) - \lambda_i(k+1,k) =   f_i(k).
\end{align*}
Inserting the fixed values of $\rho$ yields

\begin{equation}\label{eq:BIC_virtual_values}
     \psi(\vec{v}) \geq f(\vec{v}) \Bigg[ \varphi_i (k) +\frac{ \tau_i(k,k-1)}{f_i(k)} -  \frac{\tau_i(k+1,k)}{f_i(k)} \Bigg].
\end{equation}
Perhaps surprisingly, although we consider BIC truthfulness, this yields the
exact same virtual values as in the DSIC setting, but due to the fewer
monotonicity constraints there is also a reduced number of $\tau$ variables in
the dual. We dualize once more to return to the primal setting again and
essentially obtain a discrete version of Myerson's famous Lemma
\cite[Lem.~3]{Myerson1981a} in LP form:

\begin{align*}
	\max \quad &   \sum_{\vec{v} \in \vec{V}} \sum_{i=1}^{n} a_i( \vec{v}) \varphi_i(k) f(\vec{v}) \tag{\textbf{BLP2}}\label{LP:BLP2} \\ 
	\text{s.t.} \quad&  p_i(k, \vec{k}_{-i}) = v_{i,k} a_i(k, \vec{k}_{-i}) - \sum_{l=1}^{k-1} (v_{i,l+1} - v_{i,l} ) a_i(l, \vec{k}_{-i}),&& \textcolor{blue}{[\rho_i(k, \vec{k}_{-i})]}\\
	&\; \text{for} \; i \in [n], k \in [K_i], \vec{k}_{-i} \in \vec{K}_{-i}, \\
	&\sum_{\vec{v}_{-i}} \bigl[  a_i(k, \vec{k}_{-i}) - a_i(k-1, \vec{k}_{-i})  \bigr] f_{-i}(\vec{v}_{-i}) \geq 0  ,&& \textcolor{blue}{[\tau_i(k,k-1)]}\\
	&\; \text{for} \; i \in [n], k \in [K_i], \vec{k}_{-i} \in \vec{K}_{-i}, \\
	&\sum_{i=1}^{n} a_i( \vec{v}) \leq 1, && \textcolor{blue}{[\psi(\vec{v})]} \\
	&\; \text{for} \; \vec{v} \in \vec{V}.
\end{align*}
Clearly, any feasible solution of \eqref{LP:BLP2} is feasible for
\eqref{LP:BLP1} and as fixing the upward $\lambda$ variables in the dual to zero
can only increase the optimal objective value, optimality transfers from
\eqref{LP:BLP2} to \eqref{LP:BLP1} as well. Furthermore, as BIC truthfulness is
a relaxation of DSIC, any feasible solution of \eqref{LP:LP2} is feasible for
\eqref{LP:BLP2}. Whether optimality transfers as well is not obvious but can be
verified if the pair of optimal primal and dual solutions of \eqref{LP:LP2} and
\eqref{LP:DP2} is feasible for \eqref{LP:BLP2} and \eqref{LP:BDP1}. Therefore,
we will now focus on the dual problems and the $\tau$ variables, to which we
have paid little attention so far.

\subsection{Ironing}
\label{sec:ironing}

In this section we examine the connection between the primal and dual via strong
duality and what insights this can give us about optimal solutions, i.e.,
optimal auctions. The parallels between the dominant and the Bayesian setting
are striking. The primal LPs \eqref{LP:LP2} and \eqref{LP:BLP2} only differ in
the monotonicity constraints. Analogously, in the dual programs the corresponding
$\tau$ variables differ accordingly. 

\paragraph{Thought experiment}\label{page:thought-experiment} Let's temporarily fix all $\tau$ variables to zero such that both settings have the same
dual inequalities. Then the optimal dual solution would clearly be to point-wise
set
$$\psi(\vec{v}) = f(\vec{v}) \max_{i \in [n]} \varphi_i^+ (k),$$ where
$\varphi_i^+ (k)\coloneqq \max\ssets{0,\varphi_i (k)}$ is the non-negative part
of $\varphi_i (k)$. In words, given a bid profile $\vec{v}$, $\psi$ would be set
equal to the highest non-negative virtual value among all players multiplied
with the probability that this bid profile is realized. As all players share the
same $f(\vec{v})$ a winning bidder has to have the highest non-negative virtual
value. More precise, we draw the connection to the primal. By complementary
slackness, we would have the following implications for any optimal solution:

\begin{itemize}
    \item $a_i(k,\vec{k}_{-i}) > 0 \then  \psi(\vec{v}) = f(\vec{v}) \varphi_i
    (k)$, i.e. if a player has positive probability to win, then the dual
    inequality has to be met with equality, and
    \item $\psi(\vec{v}) > 0 \then  \sum_{i=1}^{n} a_i( \vec{v}) = 1$, if at
    least one player has a positive virtual value, the item is allocated with
    full probability, potentially distributed over multiple players.
\end{itemize}

Now we consider a specific bid profile $\vec{v}$ and a player's monotonicity
constraint for the variables $a_i(k,\vec{k}_{-i})$ and $a_i(k+1,\vec{k}_{-i})$.
Assume that $a_i(k,\vec{k}_{-i})>0$, then the same has to hold for
$a_i(k+1,\vec{k}_{-i})$. By complementary slackness for $a_i(k,\vec{k}_{-i})>0$
player $i$ has to be the player with the highest non-negative virtual value in
$\vec{v}$, potentially among others. Now keep all other players fixed but
increase the value of player $i$ to the next value. Although,
$a_i(k+1,\vec{k}_{-i})>0$ implies that player $i$ still has to have the highest
non-negative virtual value, if $\varphi_i (k) > \varphi_i (k+1)$ this can no
longer be ensured in general. Thus, decreasing virtual values can cause a
contradiction and the $\tau$ variables cannot be set to zero in general.

\smallskip

To resolve this problem, we find values for $\tau$ that absorb any decrease of
the virtual values, i.e., \emph{iron out} such intervals. If $\varphi$ and
$\tau$ combined are non-decreasing, we show in
\cref{lemma:optimal_deterministic_auction} that the optimal auction can easily
be found via complementary slackness. Note, that the discussion of this problem
is independent of DSIC or BIC truthfulness. Its solution, i.e. the assigned
values to $\tau$ may differ in the two cases, but we show in
\cref{lemma:ironing_equivalence} that the ironed virtual values are equivalent.
For now, we use the notion of BIC and translate this to the DSIC setting later
on. In the following, we show the existence of these dual variables that ensure
monotonicity as a solution of a system of linear equations. Furthermore, we show
that this choice is unique and, therefore, the ironing is equivalent for the
DSIC and the BIC setting. Unsurprisingly, we obtain the same values as in the
ironing algorithm of~\cite{Cai2019}. 

\begin{lemma}[Ironing]\label{lemma:ironing} Let $\varphi_i (k)$ for $k \in
    [K_i]$ be a player $i$'s virtual values. Then there exist \emph{unique}
    values $\tau \in \R_+^{K_i}$ such that 
    \begin{equation}\label{eq:ironed_virtual_values}
        \tilde\varphi_i(k)\coloneqq\varphi_i (k) +\frac{ \tau_i(k,k-1)}{f_i(k)} -  \frac{\tau_i(k+1,k)}{f_i(k)}
    \end{equation}
    is non-decreasing. We call this monotone sequence \emph{ironed virtual
    values}.
\end{lemma}

\begin{proof}
    Abusing notation for the rest of the proof we fix player $i$ and drop
    the index. This is because the ironing of \eqref{eq:ironed_virtual_values}
    is independent of all other players.
    
     To check whether the sequence $\varphi(k)$ is monotone, we construct a
     piecewise linear function whose derivative assumes the values $\varphi(k)$.
     Let $S_0 \coloneqq 0$ and $S_k \coloneqq \sum_{j=1}^k \varphi(j)f(j)$ for
     $k\in [K]$, as well as $F_0 \coloneqq 0$ and $F_k \coloneqq F(k) =
     \sum_{j=1}^k f(j)$ for $k\in [K]$. We then simply connect all $(F_{k},S_k)$
     to construct the piecewise linear function $S(x): [0,1]\map \R$, and obtain
     $\varphi(k) = \frac{S_k - S_{k-1}}{F_k - F_{k-1}}$, i.e. the slope of
     $S(x)$ for $x\in (F_{k-1},F_k)$. If this function is convex, $\varphi(k)$
     is non-decreasing and nothing needs to be done, i.e. all $\tau(k,k-1)$ can
     be set to zero. Otherwise, we construct the convex hull $H$ of this
     function. If the convex hull connects two points $(F_{l-1},S_{l-1})$ and
     $(F_{r},S_r)$ with $l<r$ and $H(x)<S(x)$ for all $x\in(F_{l-1},F_r)$, the
     interval $[F_{l},F_r]$ has to be ironed. For that we need to choose the
     $\tau$ variables such that $\tilde{\varphi}(k)$ is constant for $k \in
     [l,r]$ and leave all other virtual values unchanged. This means that, for
     all $k \in [l,r]$, it has to hold
    $$\varphi(k) +\frac{ \tau(k,k-1)}{f(k)} -  \frac{\tau(k+1,k)}{f(k)} = c$$
    for some constant $c\in\R$. This constant as well as the values of $\tau$
    can be computed via the system of linear equations
    $$\begin{pmatrix}
    1       & \frac{1}{f({l})}         & 0                      &\cdots &0\\
    \vdots  & -\frac{1}{f({l+1})}      & \frac{1}{f({l+1})}   &\ddots &\vdots\\
    \vdots  &0                           & \ddots                 &\ddots &0\\
    \vdots  &\vdots                      & \ddots    &-\frac{1}{f({r-1})}    &\frac{1}{f({r-1})}\\
    1       & 0                          & \cdots    &0      &-\frac{1}{f({r})}
    \end{pmatrix} 
    \begin{pmatrix}
    c\\
    \tau({l+1},l)\\
    \vdots\\
    \tau({r},{r-1})
    \end{pmatrix} =
    \begin{pmatrix}
    \varphi ({l}) \\
    \vdots\\
    \vdots\\
    \vdots\\
    \varphi ({r}) \end{pmatrix}.$$ Note, that as we have to leave all other
    virtual values unchanged, we cannot choose $\tau(l,l-1)$ or $\tau(r+1,r)$
    non-zero. Furthermore, if a variable $\tau(k,k-1)$ for some $k\in [l+1,r]$
    would assume the value zero, this variable separates two adjacent intervals
    $[l,k-1]$ and $[k,r]$, and these have to be ironed individually. Thus, all
    $\tau$ variables can be assumed to be positive. The first observation is
    that the square matrix has absolute determinant value $\frac{\sum _{j=l}^r
    f(j)}{\prod _{j=l}^r f(j)}$. Therefore, the solution of the system is always
    unique. Hence, there is only one option on how to choose the $\tau$
    variables in order to obtain constant $\tilde{\varphi}(k)$ for $k\in [l,r]$.
    We can also compute the value of $c$ by Cramer's rule and obtain
    $c=\frac{\sum_{j=l}^r \varphi(j) f(j)}{\sum_{j=l}^r f(j)}$, the average
    virtual value. This is the same as $c=\frac{S_r-S_{l-1}}{F_r-F_{l-1}}$, i.e.
    the slope of the convex hull within the ironed interval. The average virtual
    value clearly yields the same expected virtual welfare objective function
    that could have been realized without the ironing, thus, the ironing does
    not affect optimality.

    Lastly, we have to ensure that the unique values of $\tau$ over an ironed
    interval are indeed positive. This follows directly by strong duality as
    this gives us the existence of a dual solution with non-negative $\tau$, the
    system of linear equations then additionally yields their uniqueness.
\end{proof}

As the virtual values are identical under DSIC and BIC truthfulness, their convex
hulls are the same as well. Hence, we can show that by the uniqueness of $\tau$
in \cref{lemma:ironing} the ironing has to be equivalent in both settings.

\begin{lemma}\label{lemma:ironing_equivalence} The ironed virtual values under
    DSIC truthfulness are equivalent to the BIC ironed virtual values.
    Furthermore, as the choice of $\tau$ is unique in both cases, there is a one
    to one identification between the DSIC and BIC $\tau$ variables.
\end{lemma}

\begin{proof}
    We start with the crucial observation that, as the virtual values are the
    same in both settings, the $\tau$ variables in the DSIC ironing have to
    construct the exact same convex hull. Therefore, for any player $i$, the following equality
    must hold for all $k\in [K_i]$:
    $$\varphi_i (k) +\frac{ \tau_i(k,k-1)}{f_i(k)} -
    \frac{\tau_i(k+1,k)}{f_i(k)} = \varphi_i (k)
    +\frac{\tau_i(k,k-1,\vec{k}_{-i})}{f(\vec{v})} -
    \frac{\tau_i(k+1,k,\vec{k}_{-i})}{f(\vec{v})}.$$ \cref{lemma:ironing} gives
    us the uniqueness of the $\tau$ variables for the BIC ironing. To obtain the
    equality of the ironed values we have to define uniquely
    \begin{equation}\label{eq:unique_tau}
        \tau_i(k+1,k,\vec{k}_{-i}) \coloneqq f(\vec{k}_{-i}) \tau_i(k,k-1) \quad \text{or} \quad \tau_i(k,k-1) \coloneqq \frac{\tau_i(k+1,k,\vec{k}_{-i})}{f(\vec{k}_{-i})} .
    \end{equation}
    Note, that the latter is still valid since all
    $\frac{\tau_i(k+1,k,\vec{k}_{-i})}{f(\vec{k}_{-i})}$ have to be the same
    regardless of $\vec{v}_{-i}.$ The uniqueness follows by the non-negativity
    of the $\tau$ variables.
\end{proof}

\cref{lemma:ironing_equivalence} establishes an equivalence involving the $\tau$
variables in the ironing procedure. These are the only variables that differ
between the DSIC and BIC truthfulness settings. This allows us to connect the
two perspectives, proving Point~\ref{item:optimal-single_item-BICvsDIC}
of\cref{th:discrete-optimal-myerson}.

\begin{lemma}[DSIC Optimality]\label{lemma:DSIC_without_loss} Let
	$(\vec{a},\vec{p})$ be an optimal DSIC auction, i.e. an optimal solution of
	\eqref{LP:LP2}. Then $(\vec{a},\vec{p})$ is an optimal BIC auction, i.e.
	optimal for \eqref{LP:BLP2}.
\end{lemma}

\begin{proof}
    Let $(\vec{a},\vec{p})$ be an optimal solution of \eqref{LP:LP2}. Then there
    exists a corresponding dual solution $(\lambda,\psi,\tau)$ with values for
    $\tau$ such that the ironed virtual values in the dual are non-decreasing
    and $\lambda$ is fixed by setting all upward variables equal to zero. As the
    BIC truthfulness constraints are a relaxation of DSIC truthfulness,
    $(\vec{a},\vec{p})$ clearly is a BIC auction (see \cref{sec:truthfulness}),
    i.e., feasible for \eqref{LP:BLP2}. 
    
    If we find dual variables that are feasible for~\eqref{LP:BDP1} and satisfy
    complementary slackness, we prove that $(\vec{a},\vec{p})$ is also an optimal
    BIC auction. This can be done straightforward: The dual variables $\psi$ are
    chosen to be the same, the $\tau$ variables according to
    \eqref{eq:unique_tau} and the $\lambda$ variables are uniquely determined by
    setting all upward $\lambda$ equal to zero as well. Since all dual
    constraints reduce to being exactly the same, the complementary slackness
    immediately holds.
\end{proof}

Without loss of generality, by \cref{lemma:determinism_TU} we can assume, that in
an optimal DSIC auction $\vec{a}$ is binary since any extreme point of the set
of feasible solutions for \eqref{LP:LP2} is integer in the allocation
components. By that, any fractional optimal solution is only a convex
combination of such extreme points. The same transfers to the BIC setting: If
there are multiple integer solution optimal DSIC auctions, each of them is also
an optimal BIC auction and so is any convex combination.

\begin{lemma}[Optimality of
    Determinism]\label{lemma:optimal_deterministic_auction} The set of optimal
    BIC auctions always contains a deterministic auction, and it can be computed
    by solving the linear program \eqref{LP:LP2}.
\end{lemma}

Beyond formally ensuring the existence of a deterministic optimal solution, we
want to derive the explicit auction when we are given a bid profile $\vec{v}\in
\vec{V}$. By the complementary slackness condition 
$$a_i(k,\vec{k}_{-i}) > 0 \quad\then\quad  \psi(\vec{v}) = f(\vec{v})
\tilde{\varphi}_i (k) = f(\vec{v}) \max_{i \in [n]} \tilde{\varphi}_i (k),$$ and
the existence of an integral solution, we essentially have shown
Point~\ref{item:optimal-myerson-discrete-explicit} of
\cref{th:discrete-optimal-myerson}, i.e., receiving a bid profile $\vec{v} \in
\vec{V}$ the item is allocated to the highest non-negative ironed virtual
bidder. The corresponding payments are computed
via~\eqref{eq:first_payment_rule} which by determinism reduces to the critical
bid, i.e., the threshold value of such that the player still wins.
\begin{remark}
    \label{remark:monotonicity-free-myerson}
    Although, we can describe the optimal single-item auction clear and explicit
    via complementary slackness, an interesting insight is worth to be
    mentioned: By the ironing, i.e. the monotone $\tilde{\varphi}$, in
    combination with a deterministic tie-breaking rule we get the monotonicity
    of the primal program \emph{for free}. This is because a player with
    non-decreasing ironed virtual values can only be allocated more when
    increasing the own value while all other players stay the same, as long as
    ties are broken consistently. We can then state the optimal single-item
    auction problem as finding an allocation from
    \begin{align*}
	\max \quad &   \sum_{\vec{v} \in \vec{V}} \sum_{i=1}^{n} a_i( \vec{v}) \tilde{\varphi}_i(k) f(\vec{v}) \\ 
	\text{s.t.} \quad&  \sum_{i=1}^{n} a_i( \vec{v}) \leq 1, \\
	&\; \text{for} \; \vec{v} \in \vec{V}.
    \end{align*}
    and only have to ensure the deterministic tie-breaking rule and that the
    payments are computed via the formula \eqref{eq:first_payment_rule}.  
\end{remark}

\section{General Single-Parameter Auction Design: a KKT Approach}
\label{sec:KKT}

In general, single-parameter auctions go far beyond the single-item case. In
this section we generalize our formulation from the previous section and present
a framework for a wider range of feasibility spaces. In fact, the specialization
on the single-item setting emerges solely from the feasibility
constraints~\eqref{eq:feasibility_unity-simplex}. In a more general
single-parameter setting, we want to relax feasibility while still holding on to
truthfulness, i.e., that the players have no incentive to misreport their true
values. We maintain the linearity of the truthfulness constraints that arises
from the definition of a player's utility~\eqref{eq:utility_ex-post_definition},
which is natural for the single-parameter auction design. Our framework which
unites the techniques from the single-item setting, i.e., the duality approach
connected by complementary slackness, is a KKT system formulation
\cite{kkt1951}. 

Again we summarize our results of this section in a main theorem:

\begin{mdframed}[backgroundcolor=gray!30]
    \vspace{-8pt}
\begin{theorem}[Optimal Single-Parameter Auction]
    \label{th:discrete-optimal-general}
    For any discrete convex single-parame- ter auction setting, under the
    objective of maximizing a linear combination of revenue and social welfare
    (see~\ref{OP:general_model}), the following hold:
    \begin{enumerate}
        \item\label{item:optimal-single_parameter-general-integrality} If our
        setting is TDI, then there exists an optimal auction which is
        integral.\footnotemark 
        \item \label{item:optimal-single_parameter-general-BICvsDIC} \emph{Any}
        optimal DSIC auction is an optimal BIC auction.
        \item\label{item:optimal-single_parameter-general-explicit} The
        following DSIC auction is optimal (even within the class of BIC
        auctions):
        \begin{itemize}
        \item Pointwise choose an allocation that maximizes the generalized
        ironed virtual welfare~\eqref{OP:general_virtual_welfare_maximization},
        breaking ties arbitrarily.
        \item Collect from the allocated bidders a payment equal to their
        critical bids~\eqref{eq:payment_rule_matrix}.
        \end{itemize}
    \end{enumerate}
\end{theorem}
\end{mdframed}
\footnotetext{Recall the definition of an integral auction
from~\cref{sec:model-notation}, Page~\pageref{page:integral-auction}. The
definition of a totally dual integral (single-parameter) auction setting can be
found in~\cref{def:TU_auction_setting}.}

The framework we will present in \cref{sec:KKT-model} allows us to assume that
any feasible solution of the KKT system is also an optimal solution. Within this
rather abstract formulation, we are free to leave the ambiguity whether to
interpret the truthfulness constraint as DSIC or BIC. This not only reveals the
strong similarity of the two interpretations, but also allows us great clearness
when investigating their connection for
Point~\ref{item:optimal-single_parameter-general-BICvsDIC}. Motivated by this, in
\cref{sec:integral-auctions-KKT} we establish a setting where we can guarantee
that the optimal auction is integral and randomization or fractional allocation
is not necessary. Even in the very general case of \cref{sec:KKT},
Point~\ref{item:optimal-single_parameter-general-explicit} gives a description
of the optimal auction. We not only are able to maintain the transition to
welfare maximization (see \cref{sec:general-virtual-welfare-maximization}), but
also derive the identical payment rule as in the single-item setting. Although,
complementary slackness cannot guarantee such a clear optimal auction as in
\cref{th:discrete-optimal-myerson} in \cref{sec:application} we present an
application to show that even in a general case with combinatorial feasibility
constraints the auction can be described nicely.
 
\subsection{Notation}
\label{sec:notation-KKT}

For the general model formulation, we want to use a notation that provides
simplicity while at the same time allows modelling very general settings. Still,
we frequently draw the connection to the single-item LP formulation such that
the reader can always recall this as a special case. In the following we will
use a unified notation: In both settings of truthfulness, DSIC and BIC, each
allocation and payment variable represents an outcome per given bid profile
$\vec{v} \in \vec{V}$ and per player $i \in [n].$ We write the allocations
$\vec{a}$ and payments $\vec{p}$ as vectors of dimension $\mathcal{N} \coloneqq
n |\vec{V}| = n \cdot K_1 \cdots K_n$. One entry is a single variable, e.g.,
$a_i(k,\vec{k}_{-i})$. We further define $\vec{f}$ as a vector of the same
dimension. Each entry is the probability that a specific bid profile $\vec{v}$
is realized, i.e., $f(\vec{v})$ corresponding to the respective allocation or
payment variables $a_i(k,\vec{k}_{-i})$ or $p_i(k,\vec{k}_{-i})$ for all players
$i\in [n]$. To remain accurate with the dimensions of the objects that represent
social and virtual welfare, we also define $\vec{\nu}$ as the quadratic
$\mathcal{N} \times \mathcal{N}$ matrix with all values and similarly
$\vec{\varphi}$ with all virtual values corresponding to player $i$'s value of
the respective allocation on the diagonal and zero elsewhere. 

\paragraph{Objective function}
With this notation, we formulate the generalized objective as a linear combination of
expected revenue and expected social welfare. Equivalently, up to scaling\footnote{A linear and a convex combination are not equivalent in terms of the objective value but in terms of the optimal solution.}, we can write the objective as a convex combination of those as
$$ \alpha \; \revenue(M) + (1-\alpha) \; \welfare(M) = \alpha \vec{f}^{\top} \vec{p}
+ (1 - \alpha) \vec{f}^{\top} \vec{\nu} \vec{a},$$ with $\vec{f} ,
\vec{a}\in\R_+^{\mathcal{N}}, \vec{p}\in\R^{\mathcal{N}}, \vec{\nu} \in
\R_+^{\mathcal{N} \times \mathcal{N}}$ and $\alpha \in \R_+$.

\paragraph{Truthfulness}
Despite the more general feasibility space, the locality of the linear
truthfulness constraints, as in \eqref{LP:LP1} and \eqref{LP:BLP1}, is maintained.
They can be expressed by matrix vector notation: Matrix $A$ contains the
coefficients of the allocation variables $\vec{a}$ and $B$ the coefficients of
the payment variables $\vec{p}$ of the upward and downward truthfulness
constraints. Matrix $M$ contains the coefficients required to model the
monotonicity constraints. Whether we consider DSIC or BIC truthfulness then
depends on the coefficients and dimensions of the matrices $A, B$ and $M$, and
we do not restrict ourselves to only one of the settings. When directly
comparing the two cases, we will use $\bar{M}$ later on to describe the
monotonicity constraints in expectation. However, we differentiate between the
downward and upward constraints using, $\downward{A}$ and $\upward{A}$. The
split constraints are then
$$ \downward{A} \vec{a} + \downward{B} \vec{p} \leq 0, \quad \upward{A} \vec{a}
+ \upward{B} \vec{p} \leq 0 \quad \text{and} \quad M \vec{a} \leq 0.$$

\paragraph{Feasibility space}
Besides the truthfulness conditions, the allocations' feasibility space
$\mathcal{A}$ is represented by a finite set of convex and continuously
differentiable constraints. We assume that for each bid profile $\vec{v}\in
\vec{V}$, there are $m \in \N$ constraints. Each constraint $g_j(\vec{a}):
\vec{V} \map \R_+ $ involves \emph{only} allocation variables corresponding to
this very bid profile. That is,
$g_j(\vec{a})=g_j(a_1(\vec{v}),a_2(\vec{v}),\dots,a_n(\vec{v}))$ for $j\in [m]$
and some $\vec{v}\in\vec{V}$. To maintain ex-post feasibility the constraints
are copied for each bid profile varying over the $\vec{v} \in \vec{V}$ such that
the total number of constraints then is $\mathcal{M} := m|\vec{V}|$. E.g., in
the single-item case $\mathcal{M} = |\vec{V}|$ and each $g_j(\vec{a})$
represents \emph{the one} feasibility constraint per fixed bid profile,
see~\eqref{eq:feasibility_unity-simplex}.

Hence, an allocation $\vec{a}$ is feasible, i.e. $\vec{a} \in \mathcal{A}$, if
and only if $g_j(\vec{a}) \leq 0$ for all $j\in [\mathcal{M}].$ In our framework
we use the notion $G$ which can be seen as a vector of the $g_j$ functions,
$$
G(\vec{a}) = \begin{pmatrix}
    g_1 (\vec{a}) \\ g_2 (\vec{a}) \\ \vdots \\g_\mathcal{M} (\vec{a})
\end{pmatrix} , \quad (G(\vec{a}))^{\top} \psi = 0 \; \Longleftrightarrow \; g_1(\vec{a}) \psi_1 = 0 , \dots, g_\mathcal{M}(\vec{a}) \psi_\mathcal{M} = 0.
$$
$\nabla G(\vec{a})$ is the corresponding Jacobian matrix of $G(\vec{a})$ where
column $i$ contains all functions' derivatives with respect to the allocation
variable of player $i$ for a given bid profile $\vec{v}$. E.g. in the
single-item case we can write the linear feasibility constraints in matrix
vector notation $G \vec{a} \leq 1$ and $\nabla G(\vec{a}) = G^\top$. Note, that
we can always hide the non-negativity of the allocations within these
constraints.

\subsection{The General KKT Formulation}
\label{sec:KKT-model}

Putting the generalized objective function, truthfulness, and the feasibility
constraints together, we derive the following General
Model~\eqref{OP:general_model}:

\begin{align*}
	\max \quad &  \alpha  \vec{f}^{\top} \vec{p} + (1 - \alpha) \vec{f}^{\top} \vec{\nu} \vec{a}\tag{\textbf{GM1}} \label{OP:general_model} \\ 
	\text{s.t.} \quad&  \downward{A} \vec{a} + \downward{B} \vec{p} \leq 0, \\
    & \upward{A} \vec{a} + \upward{B} \vec{p} \leq 0, \\
    & M \vec{a} \leq 0, \\
	& G(\vec{a}) \leq 0.
\end{align*}

According to this optimization problem the KKT system covering
\ref{KKT:primal_feasibility}, the \ref{eq:KKT_dual_constraints},
\ref{eq:KKT_dual_feasibility} and the \ref{eq:KKT_complementary_slackness}
conditions then is

\begin{align*}
	\downward{A} \vec{a} + \downward{B} \vec{p} &\leq 0 , \tag{Primal feasibility} \label{KKT:primal_feasibility}\\
	\upward{A} \vec{a} + \upward{B} \vec{p} &\leq 0 ,\\
	M \vec{a} &\leq 0 ,\\
	G(\vec{a}) &\leq 0 , \\  
 \begin{pmatrix}\downward{A}^{\top}\\\downward{B}^{\top}\end{pmatrix} \downward{\lambda} + \begin{pmatrix}\upward{A}^{\top}\\ \upward{B}^{\top}\end{pmatrix} \upward{\lambda} + \begin{pmatrix} M^{\top} \\ 0 \end{pmatrix} \tau + \begin{pmatrix} (\nabla G(\vec{a}))^{\top} \\ 0 \end{pmatrix} \psi &= \begin{pmatrix} (1 - \alpha) \vec{\nu} \vec{f}  \\ \alpha \vec{f} \end{pmatrix},\tag{Dual constraints} \label{eq:KKT_dual_constraints}\\ 
    \lambda,\tau,\psi &\geq 0,\tag{Dual feasibility}\label{eq:KKT_dual_feasibility} \\
    (\downward{A} \vec{a} + \downward{B} \vec{p})^{\top} \downward{\lambda} &= 0 ,\tag{Complementary slackness} \label{eq:KKT_complementary_slackness}\\
	(\upward{A} \vec{a} + \upward{B} \vec{p})^{\top} \upward{\lambda} &= 0 ,\\
	(M \vec{a})^{\top} \tau &= 0,\\
	(G(\vec{a}))^{\top} \psi &= 0.
\end{align*}

 Due to the convexity of all components of~\eqref{OP:general_model} any feasible
 solution of the KKT system is also optimal and vice versa. Observe, that in the
 KKT system we have equations for the~\ref{eq:KKT_dual_constraints} only instead
 of inequalities as in the LP formulation. This is due to the fact that we
 consider $\vec{a}$ and $\vec{p}$ as free variables but ensure the
 non-negativity of the allocations within the feasibility constraints $G$. The
 corresponding dual variables can always ensure equality then. However, if we
 set such a dual variable positive by complementary slackness this leads to the
 primal, i.e. the allocation variable, to be zero. E.g., in the single-item case,
 not having the highest ironed virtual value of even a negative one always
 implies zero probability of allocation. We now use this system to derive
 similar results as in \cref{sec:duality-chain-myerson}.

\begin{lemma}
    If $\alpha > 0$, the payments of the optimal auction are completely
    determined by the allocation variables and the payment rule then is exactly
    the same as in the single-item setting~\eqref{eq:first_payment_rule}.
\end{lemma}

\begin{proof}
    If $\alpha > 0$, the dual constraints corresponding to the payment variables
    are
    $${\downward{B^\top}}  \downward{\lambda} + {\upward{B^\top}}
    \upward{\lambda} = \alpha \vec{f}.$$ These are exactly the equations
    considered in \eqref{LP:DP1} and \eqref{LP:BDP1}. By
    \cref{lemma:lambda_positive} we know that all $\downward{\lambda}$ are
    strictly positive. Applying the complementary slackness condition 
    $${(\downward{A} \vec{a} + \downward{B} \vec{p})}^{\top} \downward{\lambda}
    = 0$$ we see that in any feasible KKT solution, i.e., any optimal solution,
    all downward incentive compatible conditions have to bind. This gives us the
    unique single-item payment rule~\eqref{eq:first_payment_rule}, which we also
    write in matrix vector notation as
    $$\vec{p} = C \vec{a}.$$
\end{proof}

As the payment rule in the single-item setting is a linear combination of
allocation variables, we use matrix $C$ such that the payments defined by 
    \begin{equation} \label{eq:payment_rule_matrix}
        \vec{p} = C \vec{a}
    \end{equation}
are the same as~\eqref{eq:first_payment_rule}. The fixed payments arise
regardless of DSIC or BIC truthfulness. Combining the monotonicity constraints
with the payment rule~\eqref{eq:payment_rule_matrix}, we know by
\cref{lemma:DSIC_opt} for DSIC truthfulness that the local truthfulness
constraints are redundant now. This clearly transfers to BIC as well, thus, the
local upward and downward constraints can be dropped completely from the system,
as long as we hold on the \eqref{eq:payment_rule_matrix} and add a dual variable
according to the payments. We follow the notation from the single-item case and
call this free variable $\rho$. The induced the dual constraints are then
$$\begin{pmatrix} C^{\top}\\ \unitmatrix \end{pmatrix} \rho + \begin{pmatrix}
M^{\top} \\ 0 \end{pmatrix} \tau + \begin{pmatrix} (\nabla G(\vec{a}))^{\top} \\
0 \end{pmatrix} \psi = \begin{pmatrix} (1 - \alpha) \vec{\nu} \vec{f}  \\ \alpha
\vec{f} \end{pmatrix}.$$ As $\rho$ is a free variable corresponding to equality
constraints, it does not appear in the complementary slackness conditions.
Therefore, we can use its fixed values $\rho = \alpha \vec{f}$ and insert them
explicitly. Having the identical payment rule~\eqref{eq:payment_rule_matrix} to
the single-item case we know that $C^{\top} \vec{f} = \vec{\varphi} \vec{f}$
(compare\eqref{LP:DP2} and~\eqref{eq:virtual_values_appear}), and since all
elements are linear this transfers when we multiply the equation with the scalar
$\alpha$. Note that the virtual values again are equivalent for DISC and BIC
truthfulness. This closes the analogue to our chain of dual programs in
\cref{sec:duality-chain-myerson}. 

Now we can write an equivalent KKT system for the general setting, even before
essentially having to differentiate between DSIC and BIC truthfulness.
    \begin{align*}
	\vec{p} & = C \vec{a} , \tag{General KKT}\label{eq:general_KKT}\\
	M \vec{a} &\leq 0 ,\\
	G(\vec{a}) &\leq 0 ,\\
    (\nabla G(\vec{a}))^{\top} \psi &= \alpha \vec{\varphi} \vec{f} + (1 - \alpha) \vec{\nu} \vec{f} - M^{\top} \tau,\\ 
    \tau,\psi &\geq 0 ,\\
	(M \vec{a})^{\top} \tau &= 0 ,\\
	(G(\vec{a}))^{\top} \psi &= 0.
\end{align*}

\subsection{Generalized Virtual Welfare Maximization}
\label{sec:general-virtual-welfare-maximization}

The right-hand side of the dual constraints in the~\ref{eq:general_KKT} system
is a modification of the virtual values we know from revenue maximization in
\cref{sec:Myerson}. E.g. $\alpha=1$ is pure expected revenue
maximization. However, we can combine these two objects nicely and obtain
what we call \emph{generalized virtual values}. We quickly reshape 
$$\vec{\varphi^\alpha} := \alpha \vec{\varphi} + (1 - \alpha)  \vec{\nu}$$
and obtain for player $i$ with value $v_{i,k}$ the generalized virtual value
\begin{equation}\label{eq:modified_virtual_values}
    \varphi^\alpha_i(v_{i,k}) := \alpha \, \varphi_i(v_{i,k}) + (1 - \alpha)  \, v_{i,k} = v_{i,k} - \alpha \, ( v_{i,k+1} - v_{i,k}) \frac{1-F_i(v_{i,k})}{f_i(v_{i,k})}.
\end{equation}
Now for smaller values of $\alpha$, e.g. close to 0, the generalized virtual values are more likely to be non-decreasing, still, we have to consider potential non-monotonicities in the general case as well.

Even though by complementary slackness we can not directly derive such a clear
auction as in \cref{th:discrete-optimal-myerson}, we are still interested in the
role of complementary slackness in the general case and whether it requires an
analogue of ironing. To investigate this, we carry out the same thought
experiment as in \cref{sec:ironing} (Page~\pageref{page:thought-experiment}):
Fix a value profile $\vec{v}\in \vec{V}$ as well as a player $i$ and assume that
the player has positive probability of winning, i.e. $a_i(k, \vec{k}_{-i}) > 0$,
thus, the corresponding dual constraint has to bind. If we increase the player's
value, while keeping all other players' values fixed, we run into the same problem as
in the single-item case: All other players' generalized virtual values remain
the same. If player $i$'s generalized virtual values are decreasing, in general,
this might lead to a strict inequality which possibly can only be fixed by
increasing the dual variable corresponding to the non-negativity. Due to
complementary slackness, this would lead to $a_i(k+1, \vec{k}_{-i}) = 0$,
contradicting the primal monotonicity constraints. Therefore, we also need an
analogue of ironing, i.e., a choice of values for $\tau$ such that for player
$i$ the generalized virtual values \eqref{eq:modified_virtual_values} are
non-decreasing in $k\in [K_i]$. The procedure including the uniqueness remains
the same as in \cref{lemma:ironing}, as we again flatten out certain intervals
to be constant by construction of a convex hull which is not specified for
virtual values only.

The result of inserting these unique values for $\tau$ is what we call the
\emph{ironed generalized virtual values} 
$$\tilde{\vec{\varphi}}^\alpha \vec{f} :=  \vec{\varphi^\alpha} \vec{f} - M^{\top} \tau .$$
This is, that the sequences
\begin{align*}
    &\varphi_i^\alpha(v_{i,k}) + \frac{\tau_i(k,k-1,\vec{k}_{-i})}{f(\vec{v})} - \frac{\tau_i(k+1,k,\vec{k}_{-i})}{f(\vec{v})} \tag{Generalized DSIC ironing} \label{eq:modified_DSIC_ironing}\\
    &\varphi_i^\alpha(v_{i,k}) + \frac{ \tau_i(k,k-1)}{f_i(k)} -  \frac{\tau_i(k+1,k)}{f_i(k)} \tag{Generalized BIC ironing} \label{eq:modified_BIC_ironing}
\end{align*}
not only are non-decreasing in both settings, but also assume the exact same
values for each $v_{i,k}$. This is due to the fact that the generalized virtual
values again are identical in both settings and the $\tau$ variables are unique,
see \cref{lemma:ironing}. 

The only difference between the DSIC and the BIC formulation of the
\ref{eq:general_KKT} system lies in the monotonicity constraints $M \vec{a} \leq
0$ and the corresponding dual variables $\tau$. To distinguish between the
settings we use $\bar{M}$ and $\bar{\tau}$ for the BIC case. As we again add the
identical values in~\eqref{eq:modified_DSIC_ironing}
and~\eqref{eq:modified_BIC_ironing} we can unite the generalization of
\cref{lemma:ironing_equivalence} and \cref{lemma:DSIC_without_loss}.

\begin{lemma}[DSIC Optimality] \label{lemma:DSIC_without_loss_general}
    Any feasible, hence optimal, solution of the \ref{eq:general_KKT} system under DSIC truthfulness is feasible, hence optimal, for the BIC setting.
\end{lemma}

\begin{proof}
    Recall that any primal DSIC solution is feasible for the primal BIC
    constraints. The dual variables $\psi$ are chosen to be the same in both
    cases. Lastly, since the ironing is unique in both settings there is a
    one-to-one correspondence  and we can write by the
    \eqref{eq:modified_DSIC_ironing} and \eqref{eq:modified_BIC_ironing}
    $$\frac{\tau_i(k,k-1,\vec{k}_{-i})}{f(\vec{v})} -
    \frac{\tau_i(k+1,k,\vec{k}_{-i})}{f(\vec{v})} = \frac{
    \tau_i(k,k-1)}{f_i(k)} -  \frac{\tau_i(k+1,k)}{f_i(k)}$$ which is
    equivalent, using the matrix vector notation, to $M^\top \tau = \bar{M}^\top
    \bar{\tau}$. The complementary slackness transfers as well by $0 = (M
    \vec{a})^{\top} \tau = \vec{a}^{\top} M^{\top} \tau = \vec{a}^{\top}
    \bar{M}^{\top} \bar{\tau} = (\bar{M} \vec{a})^{\top} \bar{\tau}.$ 
\end{proof}

Putting everything together, we obtain an equivalent problem formulation
for~\eqref{OP:general_model}. This is, finding a generalized optimal
single-parameter auction reduced to solving the optimization problem

\begin{align*}
	\max \quad &  \vec{f}^{\top} \tilde{\vec{\varphi}}^\alpha \vec{a}  \tag{\textbf{GM2}} \label{OP:general_virtual_welfare_maximization}\\ 
	\text{s.t.} \quad&  \vec{p} = C \vec{a} ,\\
	& G(\vec{a}) \leq 0,
\end{align*}
and additionally fixing a deterministic tie-breaking rule. By this rule and the
flattening procedure we ultimately get monotonicity \emph{for free}. As the
generalized ironed virtual values are non-decreasing, increasing a player's
value while all other players stay the same ensures monotonicity of the
allocations as long as ties are broken consistently. Therefore, when solving for
an optimal \emph{generalized ironed virtual welfare maximizing} allocation, we
can restrict ourselves to pointwise finding such a maximizer under the
feasibility constraints then. The payments are computed afterwards as a function
of the allocations by the same rule as in the single-item
case~\eqref{eq:first_payment_rule}.

\subsection{Integral Auctions}
\label{sec:integral-auctions-KKT}

The generalized single-parameter auction optimization problem formulation
\eqref{OP:general_virtual_welfare_maximization} provides great transparency on
how to find an optimal auction in the general setting, that is, to focus on the
allocations and compute the payments afterwards. The only relevant constraints
for the allocations are $G(\vec{a}) \leq 0$. In this section we want to make use
of a geometrical property that also ensures determinism of the optimal
single-item auction: As the objective points into a certain direction, a polyhedral feasibility space with integral
extreme points ensures that an optimal solution has to be assumed in one of
these. To use this in the general case, we have to assume that $G(\vec{a}) \leq
0$ is linear; then, we can write the constraints in matrix vector notation as
$G\vec{a} \leq b$. By that, we define an auction setting with sufficient
conditions for obtaining an integral auction. 

\begin{definition}[TDI auction setting]\label{def:TU_auction_setting} A
    (single-parameter) auction setting will be called \emph{totally dual integral
    (TDI)}, if the allocation's feasibility constraints are given by a TDI system.
    More precisely, if there exists a system of inequalities $G\vec{a} \leq b$ such that for all integral $c$
    $$ \min_{\psi \geq 0} \{ \, \psi^\top b \; | \;  G^\top \psi = c \} $$
    is attained by an integral vector $\psi$.
\end{definition}

As we know by \cref{lemma:DSIC_without_loss_general} that optimality transfers
from the formulation under DSIC to BIC, the TDI auction setting within the DSIC
formulation ensures the existence of an integral solution under BIC truthfulness
as well. Hence, we can state the general counterpart of
\cref{lemma:determinism_TU} for the general case, that is
Point~\ref{item:optimal-single_parameter-general-integrality} of
\cref{th:discrete-optimal-general}.

\begin{lemma}[Integral optimality]
    \label{lemma:integral-optimality}
    If our (single-parameter) auction setting is TDI, then there exists an
    optimal BIC auction which is integral.
\end{lemma}

\begin{proof}
    The set of feasible allocations for the~\ref{eq:general_KKT} system is then given by the polyhedron
    $$\{ \, \vec{a} \;| \; G\vec{a} \leq b, \,  M \vec{a} \leq 0\}.$$
    As a result of our analysis we can reduce the optimization problem to~\eqref{OP:general_virtual_welfare_maximization} such that we can drop the monotonicity constraints and maximize the generalized ironed virtual welfare pointwise. 
    Therefore, the property of $G\vec{a} \leq b$ being a TDI system is sufficient for the feasibility space to be integral, which concludes the proof.
\end{proof}

This implies that several single-parameter auction settings have integral
solutions also in the BIC setting. Examples are, of course, the single-item
auction, where the TU constraint matrix even satisfies a stronger argument than TDI. But also the $k$-unit auction, the digital good auction, and in general combinatorial auctions where the constraints can be described via a TDI system $G\vec{a}\leq b$, or even stronger a TU matrix $G$ and an integral vector $b$.

We will dive deeper in a combinatorial auction in the application presented
in~\cref{sec:application}. 

\section{Application: Buying Flows on a Tree}\label{sec:application}

We consider a problem of $n$ requests to transport specific amounts of a good
through a capacitated network modelled by a directed graph. This problem is
motivated by the transport of gas where the pressure along a path is decreasing.
By this decrease of pressure, the assumption to have a directed graph without
circuits seems reasonable. In fact, a close real-world approximation of the
Greek gas network \cite{GasLib} underlying this topology. 

The network is therefore considered as the graph $G=(V,E,c)$ with nodes $V$,
directed edges $E$ and edge capacities $c_e \in \R_+$ with $c_e > 0$ attached to
each edge $e\in E$. Now there are $n$ bidders, each one having a request to send
demand $d_i$ of a good from a specific source to a sink node $s_i,t_i \in V$.
Each bidder has a private value $v_i \sim F_i$ for sending this amount over the
path $s_i\leadsto t_i$ though the network. We want to model this problem and
find an optimal mechanism for this setting.

\paragraph{Objective function} We allow a convex combination of revenue and
social welfare in the objective, i.e. for $\alpha \in (0,1)$ we maximize
$$\expect[\vec{v}\sim \vec{F}]{\alpha \sum_i p_i(\vec v) + (1 - \alpha) \sum_i v_i a_i(\vec v) }.$$

\paragraph{Feasibility space} The feasibility constraints of the network can be
modelled via the linear constraints $Ga \leq b$: Per edge $e \in E$ and per bid
profile $\vec{v} \in \vec{V}$ we have the capacity constraints
\begin{equation}\label{eq:application_constraints}
\sum_{i:\, e\in (s_i\leadsto t_i)} d_i a_i(\vec{v}) \leq c_e \quad \Longleftrightarrow \quad \sum_{i:\, e\in (s_i\leadsto t_i)} \frac{d_i}{c_e} a_i(\vec{v}) \leq 1.
\end{equation}
The corresponding non-negative dual variable for the edge capacity constraint
for $e\in E$ is denoted by $\eta_e(\vec{v})$ and we will interpret them as
\emph{edge costs}. For simplicity, we stick with the normalized
capacity constraints, i.e. the right-hand side formulation of
\eqref{eq:application_constraints}, and state this again as matrix vector
constraints by $E \vec{a} \leq 1$. Furthermore, we allow an allocation of at
most one, i.e., $a_i(\vec{v}) \leq 1$ and denote the corresponding non-negative
dual variable by $\psi_i(\vec{v}).$ Quickly observe, that the single-item
setting can be modelled by a single edge with capacity one and each player
having a demand of one to send over this edge.

\smallskip

For the remainder of this section we assume that the ironing has already been carried out, hence, the solution for
the flow auction can be found by computing an optimal solution of
\begin{align*}
	\max \quad &  \vec{f}^{\top} \vec{\tilde{\varphi}}^\alpha \vec{a}  \\ 
	\text{s.t.} \quad&  \vec{p} = C \vec{a} ,\\
	& \vec{a} \leq 1,\\
    & E \vec{a} \leq 1,
\end{align*}
including a consistent tie-breaking rule to ensure monotonicity. To actually
find the optimal auction, we investigate the complementary slackness condition of
the dual constraints associated with the allocation variables. The dual
constraint for variable $a_i(k, \vec{k}_{-i})$ with the already ironed (unique
values for $\tau$ inserted) is

\begin{equation}\label{eq:dual_flow_constraint}
    \psi_i (k, \vec{k}_{-i}) + \sum_{e\in (s_i\leadsto t_i)} \frac{d_i}{c_e} \eta_e(k, \vec{k}_{-i}) \geq f(\vec{v}) \tilde{\varphi}_i^\alpha (k).
\end{equation}
The direct implications of complementary slackness for optimal primal and dual values are:
\begin{itemize}
    \item $\tilde{\varphi}_i^\alpha (k) < 0 \then a_i(k, \vec{k}_{-i}) = 0$,
    i.e., a player with negative ironed $\alpha$-virtual value never wins,
    \item $\eta_e(\vec{v}) > 0 \then \sum_{i: e\in (s_i\leadsto t_i)} d_i a_i(\vec{v}) = c_e$, i.e., if an optimal solution requires costs on edge $e$, its capacity is fully utilized.
    \item $\psi_i(k, \vec{k}_{-i}) > 0 \then a_i(k, \vec{k}_{-i}) = 1$, i.e., there exists a condition that leads to full allocation.
\end{itemize}
If feasible, it would clearly be revenue maximizing to let all players with
non-negative ironed $\alpha$-virtual values win. We are therefore interested in
resolving conflicts in competition between bidders on an edge that would exceed
capacity when fully awarding all players. For this analysis, we need to define
and quantify the term of competition.

\begin{definition}[Competitive edge] An edge $e\in E$ will be called
\emph{competitive} with respect to a given subset of players $\mathcal{I}
\subseteq [n]$ if  
     \begin{equation}
         \label{eq:competitive-edge}
         \sum_{i \in \mathcal{I},\, e \in (s_i\leadsto t_i)} \frac{d_i}{c_e} > 1.
     \end{equation} 
     The quantity on the left-hand side of~\eqref{eq:competitive-edge} is called
     the \emph{competition} of $e$ (with respect to $\mathcal{I}$).
\end{definition}

We now present an algorithm that finds the optimal allocation, for which we
compute the payments afterwards. When fixing the received bid profile $\vec{v}$,
in the first step we will abuse the notation by dropping the values and indices
as they are consistent afterwards.

\subsection{A Combinatorial Algorithm}
\label{sec:combinatorial-algo-tree-flows}

In this section we present an algorithm that makes use of the KKT system by
using all its components, i.e. primal feasibility, the dual constraints, dual
feasibility, and complementary slackness condition and show how they interact
with each other when we are interested in finding an optimal auction. To do
this, we consider the ironed virtual values as the players' \emph{budgets} and
calculate \emph{buy-in} costs for the edges.
\begin{enumerate}
    \item Fix a $\vec{v}$ and set all corresponding $\psi_i = \eta_e = 0$ as
    well as a budget $b_i \coloneqq f(\vec{v}) \tilde{\varphi}_i^\alpha(k)$ for
    each player. Let $\mathcal{I} \subseteq [n]$ the set of all players with
    non-negative budget, i.e., non-negative ironed $\alpha$-virtual value (and
    already set for all other players $a_i(\vec{v})=0$ as they cannot afford the
    buy-in of zero). Furthermore, define the empty ordered set of players
    $\mathcal{J} \coloneqq \{\emptyset\}$.\label{alg:set}
    \item Repeat the following until there is no competitive edge left: \label{alg:competitive}
    \begin{enumerate}
        \item Find the edge $e$ with highest competition w.r.t. $\mathcal{I}$
        \item Define the buy-in for the edge as
        $$\min_{i \in \mathcal{I}, e\in (s_i\leadsto t_i)} b_i \frac{c_e}{d_i}$$
        and add this value to $\eta_e$.
        \item Charge every $i \in \mathcal{I}$ with $e\in (s_i\leadsto t_i)$ the
        buy-in $\frac{d_i}{c_e}\eta_e$ from their budget. Delete all players with $b_i=0$ from
        $\mathcal{I}$ and add them to the start of $\mathcal{J}$ by a
        deterministic rule, e.g., in lexicographical order.
    \end{enumerate}
    \item Set all $\psi_i \coloneqq \max \{ b_i , 0 \}$ and allocate in the following way:\label{alg:allocations}
    \begin{enumerate}
        \item We \emph{fully} allocate all players $i$ with $\psi_i > 0$, that
        is, we set $a_i(\vec{v})=1$. These players have endured all competition,
        and even have some of their budget left.
        \item Players in $\mathcal{J}$ \emph{fractionally} fill up the leftover
        capacities of the edge where they spent their \emph{last} budget,
        one-by-one, according to the order of $\mathcal{J}$. More formally, 
        \begin{enumerate}
            \item choose $i$, the first player of the ordered set $\mathcal{J}$,
            and $e$, the edge where $i$ spent the last budget and was added to
            $\mathcal{J}$, furthermore define $\delta_e$ to be the remaining
            capacity on edge $e$ after the prior full allocations,
            \item we set $a_i(\vec{v}) = \frac{\delta_e}{c_e}$, remove $i$ from
            $\mathcal{J}$, and go back to the previous step.
        \end{enumerate}
        \item All other players cannot afford the required edge prices and lose, i.e., $a_i(\vec{v})=0$.
    \end{enumerate}
    \item The payments are computed afterwards via the known payment rule~\eqref{eq:first_payment_rule}.
\end{enumerate}

\begin{theorem}
    The allocations and payments computed in the combinatorial algorithm are
    optimal for the flow auction problem.
\end{theorem}

\begin{proof}
    We prove optimality as the computed solution of the algorithm is feasible
    for the respective KKT system. To do so we have to verify that each
    component of the \ref{eq:general_KKT} system is satisfied. Note, that by the
    observations of~\cref{sec:general-virtual-welfare-maximization} using the
    ironed $\alpha$-virtual values and a deterministic tie-breaking rule, we get
    the monotonicity for free and do not have to consider these constraints nor
    the corresponding complementary slackness condition.
    \begin{itemize}
        \item Primal feasibility:\\
        The allocations computed in Step~\ref{alg:allocations} are clearly
        non-negative and at most one. This step also ensures the feasible flows,
        i.e. $E\vec{a} \leq 1$: As the algorithm eliminates players from
        $\mathcal{I}$ until no edge is competitive any more, i.e. until
        $$\sum_{i \in \mathcal{I},\, e \in (s_i\leadsto t_i)} \frac{d_i}{c_e}
        \leq 1$$ holds for all edges, the full allocation never exceeds
        capacity. Filling up the remaining capacities fractionally, leads to
        exploitation but never to an overflow.
        \item Dual constraints:\\
        The dual constraint corresponding to $a_i(k,\vec{k}_{-i})$ essentially
        has the form of \eqref{eq:dual_flow_constraint}. In the KKT formulation.
        there is an additional non-negative dual variable $\kappa_i (\vec{v})$
        corresponding to the non-negativity of the allocation variable.
        Furthermore, the constraint then is an equation, namely
        $$\psi_i (\vec{v}) + \sum_{e\in (s_i\leadsto t_i)} \frac{d_i}{c_e}
        \eta_e(\vec{v}) -\kappa_i(\vec{v}) = f(\vec{v}) \tilde{\varphi}_i^\alpha
        (k).$$ If a player's buy-in costs are too high, i.e., $\sum_{e\in
        (s_i\leadsto t_i)} \frac{d_i}{c_e} \eta_e(\vec{v}) > f(\vec{v})
        \tilde{\varphi}_i^\alpha (k)$ we can choose $\kappa_i(\vec{v})>0$ such
        that equality is ensured. The other way round, if there is budget left,
        we set $\psi_i(\vec{v})>0$ such that equality is ensured as well.
        \item Dual feasibility:\\
        In Step~\ref{alg:set}, $\psi_i (\vec{v})$ and $\eta_e(\vec{v})$ are
        initialized as zero. The edge prices $\eta_e(\vec{v})$ can only increase
        in Step~\ref{alg:competitive}, and in Step~\ref{alg:allocations} we set
        $\psi_i(\vec{v})$ equal to the remaining budget or zero. Hence, the dual
        variables are clearly non-negative.
        \item Complementary slackness: 
        \begin{itemize}
            \item $\psi_i(\vec{v})>0 \then a_i(\vec{v})=1$ and $a_i(\vec{v})<1
            \then \psi_i(\vec{v})=0$\\
            $\psi_i(\vec{v})$ is only positive, if player $i$ had some budget
            left in the end and is fully allocated. On the other hand, if player
            $i$ is not fully allocated, i.e. $a_i(\vec{v}) < 1$,
            $\psi_i(\vec{v})>0$ would be a contradiction to
            Step~\ref{alg:allocations}.
            
            \item $a_i(\vec{v})>0 \then \kappa_i(\vec{v})=0$ and
            $\kappa_i(\vec{v})>0 \then a_i(\vec{v})=0$

            A player with positive allocation $a_i(\vec{v}) > 0$, either is
            fully allocated and has positive $\psi_i(\vec{v})$ equal to the
            remaining budget or is fractionally allocated and has exactly zero
            budget left in the end, i.e. in both cases $$\psi_i (\vec{v}) +
            \sum_{e\in (s_i\leadsto t_i)} \frac{d_i}{c_e} \eta_e(\vec{v}) =
            f(\vec{v}) \tilde{\varphi}_i^\alpha (k)$$ and we can set
            $\kappa_i(\vec{v})=0$ then. If $\kappa_i(\vec{v})>0$ is necessary to
            ensure equality in the corresponding dual constraint, the player
            cannot afford the buy-in for the required edges and gets no positive
            allocation.
            \item $\eta_e(\vec{v}) > 0 \then \sum_{i,\, e \in (s_i\leadsto t_i)}
            \frac{d_i}{c_e} a_i(\vec{v}) = 1$ and $\sum_{i,\, e \in (s_i\leadsto
            t_i)} \frac{d_i}{c_e} a_i(\vec{v}) < 1 \then \eta_e(\vec{v}) =
            0$\\
            $\eta_e$ only increases in Step~\ref{alg:competitive} if $e$ is a
            competitive edge. By the fractional allocation in
            Step~\ref{alg:allocations} a player that would exceed the capacity
            when being fully allocated fills up the remaining capacity. Thus, a
            positive $\eta_e$ always corresponds to an edge with exploited
            capacity. An edge with leftover capacity was never the edge with
            highest competition in Step~\ref{alg:competitive}, thus, $\eta_e =
            0$ never increased from the initialization in Step~\ref{alg:set}.
            Otherwise, a player would have been eliminated from competition at
            one point and the remaining capacity would have been used by this
            player.
        \end{itemize}
    \end{itemize}
\end{proof}

\begin{acks}
We would like to thank Deutsche Forschungsgemeinschaft (DFG) for their support within subproject B07 of the Sonderforschungsbereich/Transregio 154 ``Mathematical Modelling, Simulation and Optimization using the Example of Gas Networks''. Also, we would like to thank an anonymous reviewer of the conference version of this paper for suggesting the generalization from TU to TDI for the characterization of integral auctions (\cref{lemma:integral-optimality}).
\end{acks}
\printbibliography

\appendix
\section*{Appendix}

\section{Local DSIC/BIC Characterization (Lemma~\ref{lemma:local-DSIC})}
\label{append:local-IC}
In this appendix we first give a formal proof of the alternative local
characterization of \ref{eq:DSIC-def} truthfulness presented
in~\cref{lemma:local-DSIC}, and then argue how it can readily be adapted to
\ref{eq:BIC-def} truthfulness as well.

\begin{proof}[Proof of~\cref{lemma:local-DSIC}]
    For the proof of \cref{lemma:local-DSIC} we first show the equivalence of the \eqref{eq:DSIC-def} conditions and the local DSIC
    constraints~\eqref{eq:local-IC-down} and~\eqref{eq:local-IC-up}. To do so, we first show that the local DSIC constraints imply monotonicity. Then the constraint that truthful bidding yields higher utility than deviating to the next but one value is implied by two local constraints. This can be done to the higher as well as to the lower values. In abuse of notation, we mainly drop the index of player $i$ and assume that any allocation of payment
    variable has as a second input argument the other players' fixed values
    $\vec{v}_{-i}$.

    \paragraph{$\then$)}
    Observe that the local DSIC constraints~\eqref{eq:local-IC-down}
    and~\eqref{eq:local-IC-up} are trivially implied by the \eqref{eq:DSIC-def}
    conditions, as they represent only are a reduced subset.

    \paragraph{$\Longleftarrow$)}
    Now assume that $(\vec a, \vec p)$ satisfies the local DSIC
    constraints~\eqref{eq:local-IC-down} and~\eqref{eq:local-IC-up}. We fix a
    player $i$, value $v \in V_i$ and some other players' values $\vec{v}_{-i}$.
    Adding up the two local constraints
    \begin{align*}
         v^+ a(v^+)-p(v^+) &\geq v^+ a(v)-p(v), \; \text{and}\\
        v a(v)-p(v) &\geq v a(v^+)-p(v^+)
    \end{align*}
    we obtain
    \begin{align}
        v^+ a(v^+) + v a(v) &\geq v^+ a(v) + v a(v^+)\\
        (v^+ - v) a(v^+) &\geq (v^+ - v) a(v).
    \end{align}
    Therefore, the local DSIC constraints~\eqref{eq:local-IC-down}
    and~\eqref{eq:local-IC-up} imply~\eqref{eq:local-IC-monotonicity}
    monotonicity.
    
    According to the downward deviation we consider the two local inequalities
    \begin{align*}
        v^+ a(v^+)-p(v^+) &\geq v^+ a(v)-p(v), \quad \text{and} \\
        v a(v)-p(v) &\geq v a(v^-)-p(v^-).
    \end{align*}
    Again we add them up and obtain by using monotonicity in the last inequality
    \begin{align*}
        v^+ a(v^+)-p(v^+) +v a(v)-p(v) &\geq v^+ a(v)-p(v) + v a(v^-)-p(v^-) \\
        v^+ a(v^+)-p(v^+) +v a(v) &\geq v^+ a(v) + v a(v^-)-p(v^-) \\
        v^+ a(v^+)-p(v^+)  &\geq v^+ a(v) + v a(v^-)-p(v^-) - v a(v) + v^+ a(v^-) - v^+ a(v^-)\\
        v^+ a(v^+)-p(v^+)  &\geq v^+ ( a(v) -a(v^-)) - v ( a(v) - a(v^-))+ v^+ a(v^-)  -p(v^-) \\
        v^+ a(v^+)-p(v^+)  &\geq v^+ a(v^-) - p(v^-) + (v^+ -v) ( a(v) -a(v^-))\\
        v^+ a(v^+)-p(v^+)  &\geq v^+ a(v^-) - p(v^-).
    \end{align*}
    Thus, the two local downward inequalities imply the constraint of the
    deviation to the next but one value. The analogue of the upward deviation
    follows directly when switching $v^+$ and $v^-$. This clearly can be
    extended to see that the deviation to any other value, hence all
    \eqref{eq:DSIC-def} constraints are implied by the local DSIC constraints. 
\end{proof}

Now we could rewrite the proof replacing each allocation and payment variable
with the interim versions 
$$A_i(v) \coloneqq \expect[\vec{v}_{-i}\sim \vec{F}_{-i}]{a_i(v,\vec{v}_{-i})}
\qquad \text{and} \qquad P_i(v) \coloneqq \expect[\vec{v}_{-i}\sim
\vec{F}_{-i}]{p_i(v,\vec{v}_{-i})}$$ and rederive the same result for the
\eqref{eq:BIC-def} conditions, only in expectation with respect to the other
players' values.

\section{Proof of Lemma~\ref{lemma:lambda_positive}}
\label{append:lemma:lambda_positive-proof}
    To prove this, we fix a player $i$ as well as the other players' bids
    $\vec{v}_{-i}$, and proceed inductively over $k\in[K_i]$ considering the
    equality constraints corresponding to the payment variables. We know by the
    dual borderline cases~\eqref{eq:borderline-DP1} that $\lambda_i(K,K+1,
    \vec{k}_{-i}) = \lambda_i(K+1,K, \vec{k}_{-i}) = 0$. 
    
    We start the induction for $k=K_i$ and obtain
    \begin{align*}
        &\lambda_i(K_i,K_i-1,\vec{k}_{-i}) + \lambda_i(K_i,K_i+1,\vec{k}_{-i}) -   \lambda_i(K_i+1,K,\vec{k}_{-i}) -  \lambda_i(K_i-1,K,\vec{k}_{-i}) = f(K_i,\vec{k}_{-i}) \\
        &\lambda_i(K_i,K_i-1,\vec{k}_{-i}) -  \lambda_i(K_i-1,K,\vec{k}_{-i}) = f(K_i,\vec{k}_{-i}).
    \end{align*}
    As all probabilities are strictly positive, the difference of the $\lambda$
    variables is positive as well. By the non-negativity, the same has to hold
    for the downward variable $\lambda_i(K_i,K_i-1,\vec{k}_{-i}).$

    Now assume that for some $k\in[K_i]$ we have the positive difference
    $\lambda_i(k,k-1,\vec{k}_{-i}) - \lambda_i(k-1,k,\vec{k}_{-i}) > 0$ and
    consider the equality constraint for $k-1$,
    \begin{align*}
        \lambda_i(k-1,k-2,\vec{k}_{-i}) + \lambda_i(k-1,k,\vec{k}_{-i})  -   \lambda_i(k,k-1,\vec{k}_{-i}) -  \lambda_i(k-2,k-1,\vec{k}_{-i}) = f(k-1,\vec{k}_{-i}) \\
        \lambda_i(k-1,k-2,\vec{k}_{-i}) -  \lambda_i(k-2,k-1,\vec{k}_{-i}) = f(k-1,\vec{k}_{-i}) +  \lambda_i(k,k-1,\vec{k}_{-i}) - \lambda_i(k-1,k,\vec{k}_{-i}).
    \end{align*}
    Thus, the difference $\lambda_i(k-1,k-2,\vec{k}_{-i}) -
    \lambda_i(k-2,k-1,\vec{k}_{-i})>0$ is strictly positive as well and by that
    the downward $\lambda_i(k-1,k-2,\vec{k}_{-i})>0.$ This follows for all
    downward $\lambda_i$ as long as the right-hand side is positive, i.e., until
    $\lambda_i(1,0,\vec{k}_{-i})$.

\end{document}